%% file: main.tex
\documentclass{article} 
\usepackage{iclr2027_conference,times}

\input{math_commands}

\input{preamble}

\usepackage{hyperref}
\usepackage{url}

\title{Not All Mislabels Are Equal: Efficient Task-Aware Auditing with Confidence Sequences}

\author{Fangyuan Lin\thanks{Joint first author.}\ \ , Noah Krever$^*$, Jakub \v{C}ern\'y \& Lily Xu   \\
Columbia University \\
\texttt{\{fl2744,ndk2115,jakub.cerny,lily.x\}@columbia.edu}
}

\iclrfinalcopy 
\begin{document}

\maketitle

\begin{abstract}
\input{sections/abstract}
\end{abstract}

\input{sections/introduction}
\input{sections/problem}
\input{sections/auditing_linear}
\input{sections/priority_aware_sampling}

\input{sections/entropy}
\input{sections/bayesian}

\input{sections/related_work}
\input{sections/conclusion}

\subsubsection*{Acknowledgments}
We are grateful to the \href{https://www.naturalstate.org/}{Natural State} Research \& Monitoring team for providing the biodiversity monitoring data and expertise that made this work possible, and to their partners at the \href{https://www.lewa.org/}{Lewa Wildlife Conservancy}, \href{https://www.boranaconservancy.com/}{Borana Conservancy}, and \href{https://www.lolldaigaconservancy.org/}{Lolldaiga Conservancy} who supported their work.


\subsubsection*{Reproducibility Statement}
The appendices document the proofs and experimental settings. The supplementary repository provides the implementation, frozen seeds, processed Kenya, CIFAR-10N, and CleanCoNLL input datasets, saved results, and table- and figure-generation scripts. For ReImageNet, whose source data are governed by upstream access terms, we provide download and preprocessing instructions with verification checksums.

\subsubsection*{Ethics Statement}
The biodiversity study is a retrospective replay and does not independently validate any credit or payment. Released biodiversity data exclude locations, timestamps, record identifiers, and other sensitive ecological metadata. Reference labels in CIFAR-10N, ReImageNet, and CleanCoNLL are treated as designated comparison labels and are not infallible ground truth. Prospective use must prespecify the eligible label vocabulary, preserve the randomized design, and account for expert error when verified labels are not definitive.

\subsubsection*{AI Use Statement}
Generative AI tools assisted the sketching and checking of some proofs, literature searches, implementation of experiment and figure generation code, and refining this manuscript.

\bibliography{references}
\bibliographystyle{iclr2027_conference}

\clearpage
\appendix
\crefalias{section}{appendix}
\input{appendices/top_of_appendix}
\clearpage
\input{appendices/notation}
\clearpage
\input{appendices/algos_experiments}
\input{appendices/confidence_sequences_appendices}

\end{document}

%% file: math_commands.tex
\usepackage{amsmath,amsfonts,bm}

\def\eqref#1{equation~\ref{#1}}
\def\Eqref#1{Equation~\ref{#1}}

\def\1{\bm{1}}

\DeclareMathAlphabet{\mathsfit}{\encodingdefault}{\sfdefault}{m}{sl}
\SetMathAlphabet{\mathsfit}{bold}{\encodingdefault}{\sfdefault}{bx}{n}

\newcommand{\E}{\mathbb{E}}

\newcommand{\Var}{\mathrm{Var}}



%% file: preamble.tex
\usepackage{amsmath,amssymb,amsthm,bm}
\usepackage{booktabs,array,longtable,tabularx}
\usepackage{graphicx}
\usepackage{float}
\usepackage{enumitem}
\usepackage{xcolor}
\usepackage{microtype}
\usepackage{url}
\usepackage{hyperref}
\usepackage{tikz}
\usepackage{pgfplots}
\pgfplotsset{compat=1.18}
\usepackage{dsfont}
\usepackage{cleveref}
\crefname{appendix}{appendix}{Appendices}
\Crefname{appendix}{Appendix}{Appendices}
\usepackage{algorithm}
\usepackage{algorithmicx}
\usepackage{algpseudocode}
\usepackage{multirow}
\usepackage{bm}
\usepackage{xspace}
\usepackage{makecell}
\hypersetup{colorlinks=true,linkcolor=blue!60!black,citecolor=blue!60!black,urlcolor=blue!60!black}
\graphicspath{{figures/}}

\definecolor{methodblue}{RGB}{0,70,200}
\definecolor{narrativegreen}{RGB}{0,120,70}

\theoremstyle{plain}
\newtheorem{theorem}{Theorem}
\newtheorem{lemma}[theorem]{Lemma}
\newtheorem{proposition}[theorem]{Proposition}

\theoremstyle{definition}

\theoremstyle{remark}

\newcommand{\calL}{\mathcal{L}}
\newcommand{\calR}{\mathcal{R}}
\newcommand{\calG}{\mathcal{G}}

\newcommand{\Z}{\mathbb{Z}}
\newcommand{\ind}[1]{\mathds{1}\{#1\}}

\newcommand{\Dsm}{D_1}

\newcommand{\algo}{ATLAS\xspace}
\newcommand{\algofull}{Anytime-valid Task-aware Label Auditing and Sampling\xspace}

\newenvironment{proofsketch}
  {\begin{proof}[Proof sketch]}
  {\end{proof}}

%% file: sections/abstract.tex

We study efficient, adaptive auditing of categorical labels in a fixed dataset, motivated by biodiversity credit programs in which expert verification of species records is costly and downstream scores (e.g., payment) are more sensitive to errors on certain labels than on others.
Our proposed \algo framework couples anytime-valid inference with a task-aware adaptive auditing policy. Representing label corrections through a confusion matrix, we prioritize which data points to inspect next based on both label uncertainty and the influence of potential label corrections on the downstream score function.
In our main case study on Kenyan biodiversity data, task-aware auditing reduces the expert review needed to estimate both the mislabel rate and the nonlinear biodiversity score while maintaining anytime-valid guarantees. We also evaluate ATLAS on public datasets for image classification and named-entity recognition. Across our experiments with linear and nonlinear scores, ATLAS generally reaches the same precision as competing anytime-valid methods with substantially fewer audits.
When distribution-free guarantees are prohibitively costly, we develop a model-based Bayesian alternative that can further reduce the number of audits by relying on modeling assumptions. Our biodiversity case study highlights a gap between statistical validity and practical feasibility: generic methods such as Serfling's bound and existing confidence-sequence procedures are rigorous but can require prohibitively many expert reviews in practice. Together, these results show task-aware auditing can significantly reduce the number of expert audits needed while providing finite-sample, anytime-valid guarantees.

\footnote{Code: \url{https://github.com/caPP-lab/auditing}} 




%% file: sections/introduction.tex
\section{Introduction}
\label{sec:intro}

Large labeled datasets are routinely used in consequential inference tasks, from model evaluation to scientific analysis. However, the labels in these datasets are not always correct, with annotation errors documented in settings ranging from healthcare \citep{shi2024survey,radsch2023labelling} to remote sensing \citep{pelletier2017effect}. Even widely used machine-learning benchmarks contain nontrivial label noise, with one study estimating an average error rate of 3.3\% \citep{northcutt2021pervasive}. Especially in consequential settings, practitioners may want rigorous uncertainty quantification for downstream quantities computed from the \emph{true} labels, but doing so requires verifying at least some of the labels. When expert verification of an entire dataset is impractical, two natural questions arise: How many data points must be audited to certify a downstream quantity with a desired level of confidence? Which records should be audited to reduce that uncertainty as efficiently as possible?

When audits are costly, \textit{confidence sequences} are useful as they provide confidence intervals that are simultaneously valid over all sample sizes (i.e., anytime-valid), allowing the auditor to inspect data sequentially and stop adaptively without invalidating the coverage guarantee. Recent work develops such guarantees for sampling without replacement \citep{waudby2020confidence,ryu2024gambling} to certify finite-population quantities, e.g., population mean, but these methods were not designed for the precision and affordability that audits need in practice: on our data of 33{,}403 labels, the prior--posterior-ratio confidence sequence of \citet{waudby2020confidence} needs a median of 4,021 audits to certify the mislabel rate to $\pm2$ percentage points at 95.5\% confidence, against 603 for our method, over 100 paired paths. Adaptive sampling without replacement, importance weighting, and control variates already support efficient weighted-average audits \citep{shekhar2023risk}. Our question is how to make these tools effective for \textit{categorical} corrections and their downstream scores. We use Bayesian predictions of label corrections to decide where expert review would tell us the most about the downstream score. This rule approximates the variance-minimizing audit allocation, and our confidence sequence remains valid whenever the auditor stops, even if the predictions are inaccurate. We also extend the audit to structured nonlinear scores. For example, we use a linear-envelope representation of entropy to certify the Hill--Shannon biodiversity score. Empirically, our method typically reached the target precision with substantially fewer audits than competing methods (\cref{tab:errrate,fig:p2}). 

\paragraph{Motivation: Auditing for biodiversity credits.}
Biodiversity credits are purchased by companies and investors to fund nature conservation and restoration worldwide, compensating individual landowners for measurable improvements in ecosystem health \citep{wunder2025biodiversity}. The market is estimated to grow from US\$8M annually to \$69B by 2050 \citep{wef2023biodiversity}, but credits are burdened by high overhead costs and concerns about trust and verification \citep{aide2024biodiversity}. Each project may generate tens or hundreds of thousands of sensor records documenting wildlife presence and abundance, making full expert audit impractical.
Automated species classifiers reduce labeling cost but can have substantial error: one widely used classifier had an 18.2\% error rate on an out-of-sample dataset \citep{tabak2019machine}. 
Because credit payments depend on these metrics, label errors have direct financial consequences, creating a need for efficient expert auditing. In our biodiversity credit data from Kenya, the submitted labels imply a pooled Hill--Shannon diversity of 8.16, compared with only 4.94 after expert validation---a 65\% inflation in the downstream score. An aggregate mislabel rate (e.g., ``13\% of labels are incorrect'') is insufficient because the direction of an error determines its downstream consequence. Incorrectly labeling cattle as lion, for example, may falsely introduce a rare species and substantially inflate estimated richness and taxonomic diversity (and therefore payment), whereas confusing two common, closely related species may have little effect. An effective audit should therefore account not only for where errors are likely, but also for how potential corrections affect the downstream scoring function.

The same structure arises in myriad settings. In model evaluation, for example, noisy benchmark labels impact downstream quantities such as $F_1$ score, balanced accuracy, or model rankings \citep{northcutt2021pervasive}, each of which may depend nonlinearly on the corrected labels. In epidemiology, label corrections may change prevalence estimates \citep{haile2025bias}, while in content moderation they may alter policy-compliance rates.


\begin{figure*}[t]
\centering
\includegraphics[width=\textwidth]{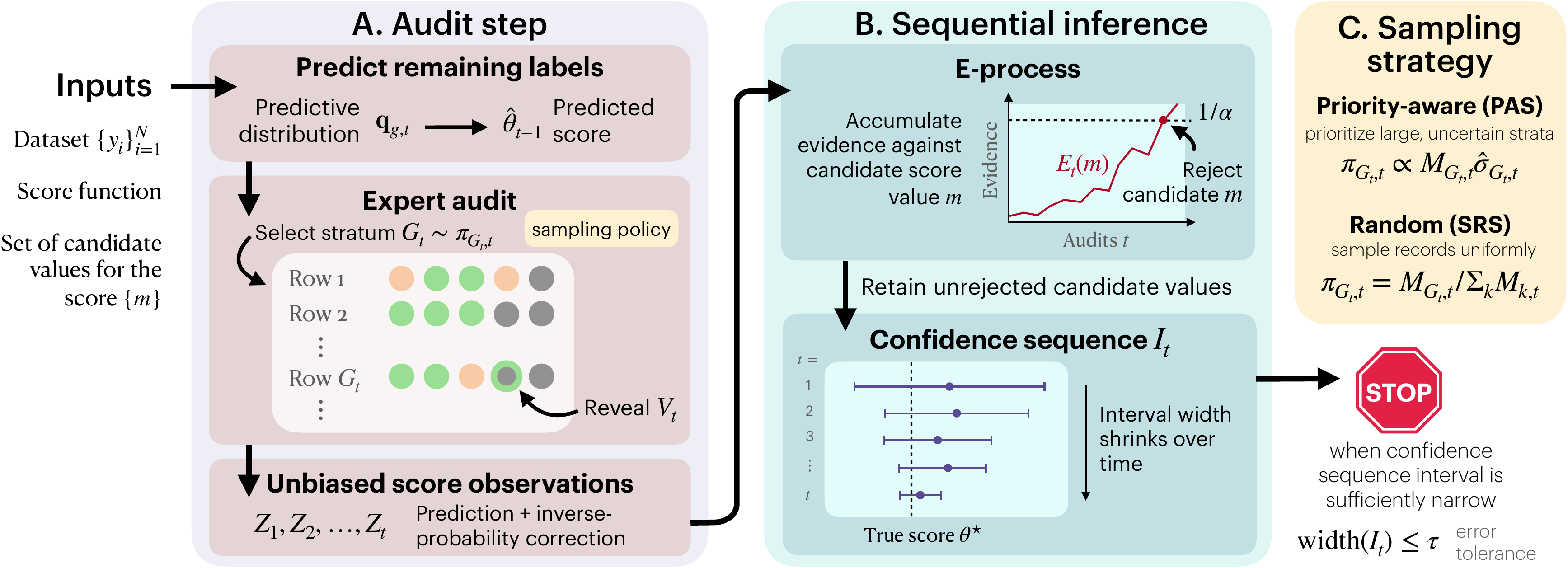}\\[-1mm]
\caption{
\textbf{Our \algo framework for \algofull.} 
\textbf{(A)}~Each step uses past audits to predict the unverified labels and the downstream score.
We select a stratum $G_t$ with probability $\pi_{G_t,t}$, then an expert audit~$V_t$ corrects the predicted score through inverse-probability weighting to produce an unbiased score observation~$Z_t$. 
\ \textbf{(B)}~E-processes accumulate evidence against each candidate score value~$m$, reject values with sufficient evidence, and use remaining values to form a confidence sequence~$I_t$, stopping once the width falls below a tolerance~$\tau$. 
\ \textbf{(C)}~
Our priority-aware sampling strategy directs reviews toward large strata with high predicted score uncertainty. We compare against random sampling as a baseline, but note that different sampling strategies only change the efficacy, not the validity of the confidence sequences.
}
\label{fig:method-overview}
\end{figure*}

To address these challenges, we make the following contributions:
\begin{itemize}[topsep=0pt, parsep=0pt, partopsep=0pt]
\item \textbf{A general framework for efficient label auditing (\cref{fig:method-overview}).}
We introduce an anytime-valid task-aware label auditing and sampling framework (\algo) for fixed categorical datasets. It uses Bayesian predictions of label corrections to guide expert review and e-processes to certify the downstream score. In our experiments, \algo typically reached the target precision with substantially fewer audits than competing anytime-valid methods.

\item \textbf{Efficient anytime-valid inference for structured nonlinear scores.} For nonlinear scores such as Hill--Shannon diversity, we exploit linear-envelope representations to build score-specific confidence sequences from our linear auditing framework. The resulting confidence sequences retain anytime-valid guarantees while supporting efficient adaptive auditing.

\item \textbf{A model-based Bayesian alternative using the same task-aware sampling strategy.}
When distribution-free guarantees require an impractically large number of audits, we retain the same
task-aware strategy but use a Bayesian model of the confusion matrix for uncertainty quantification, trading assumption-free validity for
greater statistical efficiency.
\end{itemize}


%% file: sections/problem.tex
\section{Problem Formulation}
\label{sec:setting}

A provided dataset contains $N$ records. Record $i\in[N]$ has a submitted label~$y_i$ (known) and verified label~$v_i$ (unknown until verification). Let $\calL :=\{1,\dots,L\}$ be the full normalized label vocabulary, and let $\calR := \{y_i : i\in[N]\}\subseteq\calL $ be the set of submitted labels with $R:=|\calR|$ elements. Thus $y_i\in\calR$ and $v_i\in\calL$.

\paragraph{Confusion matrix.} A scalar mislabel rate records how often $y_i\neq v_i$ across the entire dataset; it does not capture error rates across different classes. We therefore track a full \emph{confusion matrix}, which counts records for each pair of submitted and verified labels (\cref{fig:kenyan-mislabel-rate}). The diagonal represents how many records are correctly labeled.
The unknown finite-population confusion matrix is $\mathbf{C}^\star\in\Z_{\ge0}^{R\times L}$, where each entry $C^\star_{rs}:=\#\{i:y_i=r,\ v_i=s\}$ indicates the number of data points submitted with row~$r$ and verified with column~$s$. Row~$r$ contains records submitted as label~$r$, while column~$s$ contains records verified as label~$s$. The row sums $N_r:=\sum_{s\in\calL}C^\star_{rs}=\#\{i:y_i=r\}$ are known before auditing. 

\paragraph{Score functions.}
We use \emph{score function} to denote any downstream quantity computed from the corrected labels. Formally, a score function maps the submitted-to-verified confusion matrix $\mathbf{C}$ to a scalar outcome of interest, such as an error rate, model-evaluation metric, or biodiversity score. 
Importantly, we are interested in finite-population inference. When auditing labels we are only interested in the score for the particular submitted dataset, unlike in conventional superpopulation inference where we seek to estimate a parameter of an underlying distribution from which the data are assumed to be i.i.d.\ draws.
For the biodiversity credits setting, the relevant score function is challenging as it is nonlinear and nonsmooth; details in \cref{sec:biodiversity-setting}.


\paragraph{Strata: splitting rows by metadata.} Records often carry metadata, such as the source, location, device, or time period, that may be associated with different correction patterns. Splitting each row by these metadata creates smaller groups of records which we call \emph{strata}. We write $g=(r,a)$ for the stratum containing records with submitted label $r$ and metadata value~$a$. In our biodiversity application, $a$ refers to a camera station, where corrections are strongly clustered. Then $r(g)$ is the submitted label shared by records in stratum~$g$. The size of each stratum is known before any audit, and each row is the union of its strata. When no metadata is used, each row is its own stratum. Stratification is useful when correction behavior varies across metadata-defined subgroups, because it allows the audit to focus more precisely on where uncertainty remains. 



%% file: sections/auditing_linear.tex
\section{Confidence Sequences for Auditing Linear Score Functions}
\label{sec:linear}

We first explore how to turn expert audits into a \emph{confidence sequence}, illustrated in \cref{fig:method-overview}b. Confidence sequences provide uncertainty intervals that are valid after every audit, allowing the auditor to monitor uncertainty and stop as soon as desired precision is reached. In this section, we instantiate the auditing workflow in \cref{fig:method-overview} for scores that are linear in the confusion matrix. 



Let $\mathbf{D} \in \mathbb{R}^{R \times L}$ be a deterministic \emph{contribution matrix}, where the $(r,s)$ element $D_{rs} \in [0,1]$ is the contribution to the score from a data point submitted as label $r$ and verified as label~$s$ and is known in advance. For example, when the score function is the mislabel rate, $\mathbf{D}$ is a binary matrix with $D_{r,s} = \ind{r \neq s}$. We denote the $r$-th row of this matrix as $\mathbf{d}_r^\top$.


A score is linear if it can be expressed as $\theta^\star = \frac{1}{N} \sum_{i=1}^{N} D_{y_i,v_i} = \frac{1}{N} \sum_{r \in \calR} \sum_{s \in \calL} C^\star_{rs} D_{rs}$. With linear scores, each record contributes a fixed, known amount determined by its submitted and verified labels. If all verified labels were known, the true score $\theta^\star$ could be computed exactly. During the audit, however, only a subset of the data have been reviewed.

\paragraph{A starting point: Serfling's inequality.}
For a fixed number of labels sampled uniformly without replacement, Serfling's inequality \citep{serfling1974probability} sharpens the standard Hoeffding bound and gives a finite-sample, distribution-free confidence interval for the mislabel rate. Even so, on the Kenyan dataset, certifying the mislabel rate to $\pm 2$ percentage points at 95.5\% confidence requires 4,154 reviews, about one in eight records (\cref{tab:errrate,fig:kenyan-mislabel-rate}). A Wald interval reaches the same precision in 1,481 reviews, but
relies on an asymptotic approximation and, in this comparison, the true population variance. We therefore seek greater efficiency while retaining rigorous finite-sample validity.

\paragraph{Constructing an unbiased audit observation.}
We first construct a quantity~$Z_t$ whose conditional mean, given everything observed before audit~$t$, is the downstream score~$\theta^\star$:
\begin{equation}
  \mathbb E[Z_t\mid\mathcal F_{t-1}]=\theta^\star\qquad\text{for every labeling}
  \label{eq:need}
\end{equation}
where $\mathcal F_{t-1}$ denotes the audit history available before audit $t$. 
Let $\mathcal A_{t-1}$ denote the records audited so far. For each stratum $g$, let $M_{g,t}=|(g\setminus\mathcal A_{t-1}|$ be the number of unaudited records remaining. Their score contributions have mean $\mu_{g,t}=\frac1{M_{g,t}}\sum_{i\in g \setminus \mathcal{A}_{t-1}}D_{y_iv_i}$ and variance $\sigma^2_{g,t}=\frac1{M_{g,t}}\sum_{i\in g \setminus \mathcal{A}_{t-1}}(D_{y_i,v_i}-\mu_{g,t})^2$. 
The linear score can be split into:
\begin{equation}
  N\theta^\star
  =\sum_{i\in\mathcal A_{t-1}}D_{y_iv_i}+\sum_g M_{g,t}\,\mu_{g,t}.
  \label{eq:split}
\end{equation}
After $t-1$ audits, the first sum is known and the second must be learned (or estimated) by the auditor. 
Before audit~$t$, using only $\mathcal F_{t-1}$, the auditor fixes sampling probabilities $\{\pi_{g,t}\}_g$ (where $\sum_g\pi_{g,t}=1$) and a predicted verified-label distribution $\mathbf{q}_{g,t} \in \Delta^{L-1}$ for each nonempty stratum~$g$. 
Using these predictions, we predict the downstream score $\theta^\star$ via the estimator:
\begin{align}
    \hat{\theta}_{t-1} = \frac{1}{N} \Big(\sum_{i \in \mathcal{A}_{t-1}} D_{y_i,v_i} + \sum_{g} M_{g,t} \cdot  \mathbf{d}_{r(g)}^\top \mathbf{q}_{g,t} \Big) \ .
    \label{eq:score-estimator}
\end{align}
The first term uses the verified contributions of audited records, while the second predicts the contributions of the remaining records. 
The auditor then draws a stratum~$G_t$ according to
$\Pr(G_t=g\mid\mathcal F_{t-1})=\pi_{g,t}$,
selects a record $J_t$ uniformly from $G_t\setminus\mathcal A_{t-1}$, and reveals its verified label $V_t=v_{J_t}$. Thus an unaudited record in stratum $g$ is selected with probability $\pi_{g,t}/M_{g,t}$.

Because strata may be sampled nonuniformly, simply replacing the predicted contribution with the newly verified one would generally introduce bias. To update this prediction with the next audit, we therefore use an inverse-probability-weighted residual correction \citep{horvitz1952generalization}:
\begin{equation}
    Z_t = \hat{\theta}_{t-1} + \frac{M_{G_t,t}}{N \cdot \pi_{G_t,t}} \left( D_{r(G_t),V_t} - \mathbf{d}_{r(G_t)}^\top \mathbf{q}_{G_t, t} \right).
    \label{eq:Zt}
\end{equation}
The residual compares the newly verified contribution~$D_{r(G_t),V_t}$ with its prediction~$\mathbf{d}_{r(G_t)}^\top \mathbf{q}_{G_t, t}$, while the inverse-probability weight $\frac{M_{G_t,t}}{N \pi_{G_t,t}}$ corrects for nonuniform sampling. 
Consequently, $Z_t$ is an unbiased estimate for $\theta^\star$ for any predictable predictions and sampling probabilities, as formalized in \cref{lem:unbiased}. 
The prediction term acts as a control variate, analogous to \citet{sarndal2003model}, and improves efficiency while preserving validity.

\begin{lemma}[One audit, one unbiased observation]
\label{lem:unbiased}
If $\mathbf{q}_{\cdot,t}$ and $\bm{\pi}_{\cdot,t}$ are predictable and $\pi_{g,t}>0$
whenever $M_{g,t}>0$, then $\E[Z_t\mid\mathcal F_{t-1}]=\theta^\star$ for every labeling. Define the stratum-level prediction error as $b_{g,t}=\mu_{g,t}- \mathbf{d}_{r(g)}^\top \mathbf{q}_{g,t}$, then the conditional variance is:
\begin{equation}
  \Var(Z_t\mid\mathcal F_{t-1})=\frac1{N^2}\left[\sum_g\frac{M_{g,t}^2}{\pi_{g,t}}\bigl(\sigma^2_{g,t}+b_{g,t}^2\bigr)
  -\Bigl(\sum_g M_{g,t}b_{g,t}\Bigr)^2\right].
  \label{eq:varZ}
\end{equation}
\end{lemma}
Better predictions reduce variance (hence improving efficiency), while unbiasedness holds regardless of prediction quality.
In experiments, $\mathbf{q}_{g,t}$ is the posterior mean under a hierarchical Dirichlet model \citep{gelman1995bayesian}, updated using only the audits observed before time~$t$.


\paragraph{E-processes.} 
For each candidate score value $m$, we construct two e-processes that accumulate evidence against the null hypothesis $\theta^\star=m$. The confidence sequence is then obtained by retaining candidate values for which neither e-process has crossed the rejection threshold. Ville's inequality gives simultaneous coverage (i.e., anytime validity) over all audit times.

Following \citet{waudbysmith2024estimating},
let $\lambda^\pm_i(m)\geq 0$ be predictable betting stakes. Under this null hypothesis, \cref{lem:unbiased} implies that each factor below has conditional mean one:
\begin{equation}
  E^+_t(m)=\prod_{i\leq t}\bigl\{1+\lambda^+_i(m)(Z_i-m)\bigr\},\qquad
  E^-_t(m)=\prod_{i\leq t}\bigl\{1-\lambda^-_i(m)(Z_i-m)\bigr\}.
\end{equation}
Thus at the true value $m = \theta^\star$, both processes are nonnegative martingales with initial value one. 


The only requirement on the stakes $\lambda^{\pm}$ is that they are predictable and keep every factor nonnegative, for every possible revealed label. This is always achievable because $D_{rs}\in[0,1]$, so the range of $Z_t$ is known before each audit, and such stakes can be chosen. \Cref{app:linear-proofs} gives this range and the stakes used in our experiments.

\begin{theorem}[Confidence sequence for a linear score]
\label{thm:linear-cs}
Let $I_t=\{m:E^+_s(m)<2/\alpha, E^-_s(m)<2/\alpha\ \text{for all}\ s\leq t\}$. Then $\Pr\{\theta^\star\in I_t\, \forall t\}\geq1-\alpha$ for every labeling.
\end{theorem}
\begin{proofsketch}
Applying Ville's inequality (\cref{lem:ville}) to the two martingales at $m = \theta^\star$, each crosses $2/\alpha$ with probability at most $\alpha/2$, uniformly over time. 
A union bound over the two tails gives total error at most $\alpha$. Details in \cref{app:linear-proofs}.
\end{proofsketch}
The same e-process construction extends to ratios of linear scores: the recall $\rho^\star=C^\star_{ss}/\sum_rC^\star_{rs}$ of species $s$ and the $F_1$ score are handled the same way in \cref{app:scores}.

%% file: sections/priority_aware_sampling.tex
\section{Priority-Aware Sampling}
\label{sec:pas}



The confidence sequence in \cref{sec:linear} determines when an audit may stop. We now consider what to audit next. Because different label corrections can vary greatly in their impact on the downstream score, we aim to prioritize expert audits toward strata whose unresolved labels contribute most to uncertainty in that score, a process illustrated in \cref{fig:method-overview}a.

Under exact predictions ($b_{g,t} = 0$ and $\mu_{g,t} = \mathbf{d}_{r(g)}^\top \mathbf{q}_{g,t}$), \cref{eq:varZ} reduces to $\operatorname{Var}(Z_t\mid\mathcal F_{t-1})
=\frac1{N^2}\sum_g\frac{M_{g,t}^2\sigma_{g,t}^2}{\pi_{g,t}}$. This is minimized by $\pi_{g,t} \propto M_{g,t} \sigma_{g,t}$, suggesting we prioritize strata in proportion to their remaining (unaudited) size times uncertainty in their score contributions. 
\begin{proposition}[Priority-aware allocation]
\label{prop:neyman}
If the predictions are exact, \eqref{eq:varZ} is minimized subject to floors
$\pi_{g,t}\geq\ell_{g,t} > 0$ by $\pi_{g,t}=\max\{\ell_{g,t},M_{g,t}\sigma_{g,t}/\nu\}$,
where $\nu > 0$ normalizes the probabilities. Without floors, the minimum is
$N^{-2}(\sum_gM_{g,t}\sigma_{g,t})^2$, the Neyman variance
\citep{neyman1934two}.
\end{proposition}



In practice, we estimate the standard deviation $\sigma_{g,t}$ from the predicted verified-label distribution $\mathbf{q}_{g,t}$. When the scoring function is the mislabel rate, each contribution is Bernoulli. If $\hat{e}_{g,t} = \mathbf{d}_{r(g)}^\top \mathbf{q}_{g,t} = \widehat{\Pr}(\text{verified label} \neq r(g) \mid \mathcal{F}_{t-1})$ is the predicted error rate in stratum $g$, then the estimated standard deviation is $\hat{\sigma}_{g,t} = \sqrt{\hat{e}_{g,t}(1-\hat{e}_{g,t})}$. 

Then to ensure continued exploration, we mix the resulting priority allocation with sampling proportional to the number of remaining records in each stratum $g$:
\begin{equation}\label{eq:pas-allocation}
    \pi_{g,t} = \varepsilon \frac{M_{g,t}}{\sum_k M_{k,t}} + (1-\varepsilon) \frac{M_{g,t} \hat{\sigma}_{g,t}}{\sum_k M_{k,t} \hat{\sigma}_{k,t}} \ .
\end{equation}
The first term guarantees continued exploration of all remaining records, while the second approximates the variance-minimizing allocation of \cref{prop:neyman}. The end-to-end priority-aware sampling procedure is presented as \cref{alg:pas} in \cref{app:algos}.



\subsection{Experimental Results: Certifying the Mislabel Rate}
\label{sec:linear-results}
\Cref{tab:errrate} compares methods for certifying the overall mislabel rate using corrected records from the Kenya biodiversity, CIFAR-10N \citep{wei2022cifarn}, ReImageNet \citep{volkov2026imagenet}, and CleanCoNLL \citep{rucker2023cleanconll} datasets. ATLAS-PAS and ATLAS-SRS use the same confidence sequence from \cref{sec:linear,thm:linear-cs} and differ only in the selection of the next record. We also compare with Prop-MS and Prop-M + CV implementations of the risk-limiting financial audit (RLFA) methodology in \citet{shekhar2023risk}, using the authors' code, \footnote{\url{https://github.com/sshekhar17/WeightedWoRConfSeq}} the prior--posterior-ratio (PPR) confidence sequence \citep{waudby2020confidence}, and fixed-sample Wald, Serfling, and Chebyshev references.

ATLAS-PAS reduces the median audit count relative to random sampling across all four datasets and achieves the lowest median among anytime-valid methods on three. On CleanCoNLL, the RLFA control-variate baseline has a lower median, but reaches the target precision on only three-quarters of paths within budget, compared with every path for ATLAS-PAS. Among the fixed-sample references, Serfling and Chebyshev baselines require more audits than ATLAS-PAS in every setting. On ReImageNet and CleanCoNLL, the Wald baseline reaches the target precision with fewer audits, but it is an asymptotic approximation and uses the true variance.


\Cref{fig:kenyan-mislabel-rate} shows the confusion matrix of the Kenyan biodiversity dataset and the audit curves for ATLAS and each baseline, with a target precision of $\pm 2$ percentage points for the mislabel rate at a nominal confidence level of 95.5\%. For each linear score, \cref{app:additional-results} reports audit curves and stopping-time distributions. We describe the baselines in \cref{app:linear-comparators} and examine when priority-aware allocation saves the most audits through a controlled experiment (\cref{app:synthetic-experiment-settings}) and sensitivity analyses (\cref{app:sensitivity}).

\providecommand{\attaincell}[3]{%
  \begin{tabular}[t]{@{}r@{}}
  #1\\[-2pt]{\scriptsize[#2, #3]}
  \end{tabular}%
}

\begin{table*}[t]
\centering
\caption{Median audits to reach an interval half-width of 2 percentage points at 95.5\% nominal confidence over 100 paired paths. Counts are rounded to the nearest audit. A budget of 5,000 audits was used for the anytime-valid methods.}
\label{tab:errrate}
\makebox[\textwidth][c]{%
\footnotesize
\setlength{\tabcolsep}{5pt}
\renewcommand{\arraystretch}{1.15}
\begin{tabular}{@{}
>{\raggedright\arraybackslash}p{3.4cm}
>{\raggedright\arraybackslash}p{3.25cm}
>{\raggedright\arraybackslash}p{2.75cm}
>{\raggedleft\arraybackslash}p{1.45cm}
>{\raggedleft\arraybackslash}p{1.7cm}
>{\raggedleft\arraybackslash}p{1.6cm}
>{\raggedleft\arraybackslash}p{1.7cm}@{}}
\toprule
& & & \multicolumn{4}{c}{\# audits to $\pm 2$ percentage points}\\[-2pt]
& & & \multicolumn{4}{c}{\scriptsize Median [25th, 75th percentiles]}\\
\cmidrule(l){4-7}
\textbf{Method}
& \textbf{Sampling strategy}
& \textbf{Guarantee}
& \textbf{Kenya}
& \textbf{CIFAR-10N}
& \textbf{ReImageNet}
& \textbf{CleanCoNLL}\\
\midrule

\textbf{ATLAS-PAS} {\scriptsize (ours)}
& Estimated variance-minimizing
& Anytime-valid
& \attaincell{\textbf{603}}{555}{636}
& \attaincell{\textbf{790}}{725}{881}
& \attaincell{\textbf{1,394}}{1,161}{1,700}
& \attaincell{281}{242}{334}\\
\addlinespace[3pt]

ATLAS-SRS {\scriptsize (ours)}
& Uniform
& Anytime-valid
& \attaincell{814}{718}{866}
& \attaincell{980}{914}{1,056}
& \attaincell{1,539}{1,287}{1,764}
& \attaincell{293}{225}{368}\\
\addlinespace[3pt]

Prop-MS {\scriptsize (Shekhar et al.)}
& Proportional to predicted mislabel probability
& Anytime-valid
& \attaincell{$>$5,000}{$>5$k}{$>5$k}
& \attaincell{$>$5,000}{4,530}{$>5$k}
& \attaincell{1,671}{1,417}{1,994}
& \attaincell{1,575}{1,151}{1,966}\\
\addlinespace[3pt]

Prop-M + CV {\scriptsize (Shekhar et al.)}
& Uniform
& Anytime-valid
& \attaincell{2,016}{1,552}{2,233}
& \attaincell{2,404}{2,107}{2,886}
& \attaincell{2,631}{2,038}{3,495}
& \attaincell{\textbf{233}}{100}{$>5$k}\\
\addlinespace[3pt]

PPR {\scriptsize (Waudby-Smith \& Ramdas)}
& Uniform
& Anytime-valid
& \attaincell{4,021}{3,572}{4,237}
& \attaincell{4,035}{3,541}{4,310}
& \attaincell{2,216}{2,021}{2,338}
& \attaincell{435}{340}{529}\\
\addlinespace[3pt]

Wald$^\dagger$
& Uniform
& Asymptotic, fixed $n$
& 1,481 & 1,434 & 748 & 119\\
\addlinespace[3pt]

Serfling
& Uniform
& Finite-sample, fixed $n$
& 4,154 & 4,332 & 4,112 & 4,670\\
\addlinespace[3pt]

Chebyshev$^\ddagger$
& Uniform
& Finite-sample, fixed $n$
& 9,811 & 10,870 & 9,583 & 13,278\\

\bottomrule
\end{tabular}%
}
\par\vspace{3pt}
{\centering\scriptsize
Bold: smallest median among anytime-valid methods.
$\dagger$ Oracle population variance.
$\ddagger$ Analytic fixed-$n$ count using variance bound $1/4$.
\par}
\end{table*}

\begin{figure}[ht]
    \centering
\raisebox{-.5\height}{%
  \includegraphics[width=.38\linewidth]
  {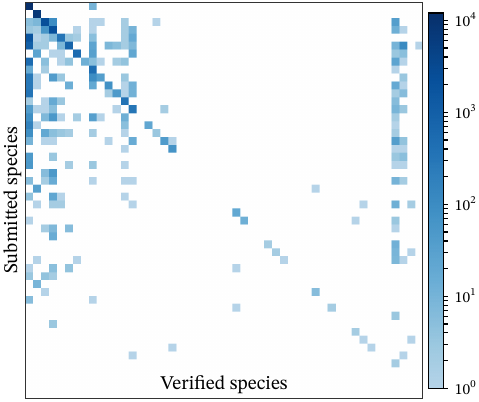}%
}
\hfill
\raisebox{-.5\height}{%
  \includegraphics[width=.60\linewidth]
  {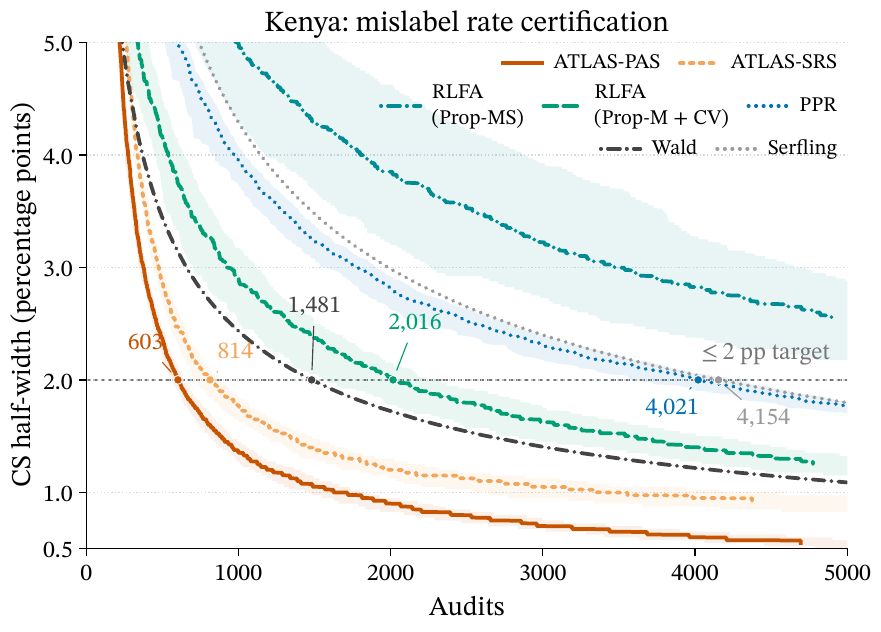}%
}
\vspace{-6pt}
    \caption{Audits on the Kenya biodiversity dataset.
\ \textbf{(Left)}~Submitted-to-verified confusion matrix (input data). Diagonal entries are correct labels, off-diagonal entries are mislabels. Color indicates counts on a log scale.%
\ \textbf{(Right)}~95.5\% confidence sequence half-width as the number of audits increases. Curves show medians over 100 audit paths with the 25th--75th percentiles. Annotated dots mark the audits to reach a two percentage point precision. Corresponding results for CIFAR-10N (\cref{app:cifar-mislabel-figure}), ReImageNet (\cref{app:reimagenet-mislabel-figure}), and CleanConLL (\cref{app:cleanconll-mislabel-figure}) are reported in \cref{tab:errrate}.}
    \label{fig:kenyan-mislabel-rate}
\end{figure}

%% file: sections/entropy.tex
\section{Auditing Nonlinear Scores via Linear Envelopes}
\label{sec:entropy}
In practice, downstream scores are often nonlinear, so the linear construction of \cref{sec:linear} cannot be applied directly with a single fixed contribution matrix. For concave scores, we show they nevertheless admit an anytime-valid audit by exploiting their representation as envelopes of linear scores: a single predictable tangent direction yields a valid upper confidence bound, while the lower bound requires ruling out every plausible label composition, which we certify by convex optimization for Shannon entropy. We develop the linear envelope construction for Shannon entropy, yielding a confidence sequence for the nonlinear Hill–Shannon biodiversity score. The same envelope principle applies to worst-three precision and worst-three taxonomic risk on CIFAR-10N. More generally, one can bound a nonlinear score by taking its range over a confidence region for the confusion matrix. We include this projection approach as a baseline in \cref{fig:p2}. We can also use the structure of a score directly, as we do below for concave scores. For the F1 score on ReImageNet and CleanCoNLL, we use its representation as a ratio of linear scores (\cref{app:scores}).
\paragraph{Concave scores are envelopes of linear scores.}
Let $\mathbf{p}^\star\in\Delta^{L-1}$ denote the true corrected label proportions, with
$p_s^\star=\frac{1}{N}\sum_r C_{rs}^\star$.
Then define $H : \Delta^{L-1} \to \mathbb{R}$ as a concave, differentiable (on its relative interior) score function of a label composition $\mathbf{p}$. For any reference composition $\mathbf{\bar p}\in\Delta^{L-1}$, concavity implies that the tangent plane to $H$ at $\mathbf{\bar p}$ upper bounds $H$. Using $\sum_sp_s=1$, this tangent plane can be written as a linear score:
\begin{equation}\label{ent:eq:tangent}
  H(\mathbf p)\ \le\ \sum_{s\in\mathcal L}p_s\,\psi_s(\mathbf{\bar p}),\qquad \text{where }
  \psi_s(\mathbf{\bar p})=H(\mathbf{\bar p})+\partial_s H(\mathbf{\bar p})-\sum_j{\bar p}_j\,\partial_jH(\mathbf{\bar p}),
\end{equation}
with equality at $\mathbf {\bar p}= \mathbf p$. 
Thus $H(\mathbf p)=\inf_{\mathbf{\bar p}}\sum_sp_s\psi_s(\mathbf{\bar p})$ is the infimum of these linear upper bounds over reference compositions $\bar{\mathbf p}$. So a record verified as
$s$ contributes the known amount $\psi_s(\mathbf{\bar p})$ to the linear score indexed by $\mathbf{\bar p}$. 
Specializing $H$ to Shannon entropy, $H(\mathbf{p})=-\sum_{s\in\mathcal L}p_s\log p_s$ gives  $\psi_s(\mathbf{\bar p})=\log(1/{\bar p}_s)$, so the tangent plane is the cross-entropy
$\mathrm{CE}(\mathbf p,\mathbf{\bar p}):=\sum_sp_s\log(1/{\bar p}_s)=H(\mathbf p)+\mathrm{KL}(\mathbf p\|\mathbf{\bar p})$. 
\Eqref{ent:eq:tangent} is therefore Gibbs' inequality $\mathrm{KL}\ge0$
\citep{cover1991elements}.

\paragraph{Upper bound: one linear audit.}
 
Any predictable reference $\mathbf {\bar p}$ defines a linear score that upper bounds $H(\mathbf{p}^\star)$; its choice affects only the confidence width but not validity. We therefore take the best guess of $\mathbf {p^\star}$ available before the audit, i.e., the plug-in composition $\mathbf{\hat p}_{t-1}$, which
counts audited records at their verified labels and unaudited records at their
predicted distributions $\mathbf{q}_{g,t}$. Applying the linear auditing procedure of \cref{sec:linear} to this linear tangent score yields the upper confidence sequence $U_t$ (\cref{app:entropy}).





\paragraph{Lower bound: ruling out every composition.}
A tangent plane lies above $H$, so a single cross-entropy cannot certify a lower bound. Let $H^\star:=H(\mathbf p^\star)$. To show $H^\star > h$, we must rule out every composition $\mathbf p$ still consistent with the audits that could yield entropy at most $h$. Thus the lower-bound problem becomes a composite hypothesis test: we must reject the entire set of feasible compositions whose entropy is at most $h$. Let $\mathcal{P}_t = \{\mathbf p \in \Delta^{L-1} : p_s \ge O_{t,s}/N\}$, where $O_{t,s}$
counts audited records verified as $s$ after $t$ audits. Following \citet{manole2023martingale}, for each candidate composition $\mathbf p$ we form an e-process $E^L_t(\mathbf p; h)$ whose factors have conditional mean at most one at $(\mathbf{p^\star}, H^\star)$, using a single predictable stake shared by all candidates. We rule out $h$ once $\min_{\mathbf p \in \mathcal{P}_t} \log E^L_t(\mathbf p; h) \ge \log (1/\alpha_-)$, where $\alpha_-$ is the error budget of the lower end. Sharing one stake and one prediction correction across all candidates makes this objective convex in $\mathbf p$, so this minimum can be certified by convex optimization rather than enumeration (\cref{app:entropy}). Splitting the total budget as $\alpha = \alpha_+ + \alpha_-$, the resulting interval $[L_t,U_t]$ forms the desired confidence sequence. The end-to-end procedure is presented as \cref{alg:entropy} in \cref{app:algos}.

\subsection{Experimental Results: Certifying Nonlinear Scores}
\label{sec:nonlinear-results}


We measure uncertainty using the \emph{relative margin of error}.
For an interval $[a,b]$, we define this as the interval half-width of the confidence sequence divided by the average of its two endpoints:
$\frac{(b-a)/2}{(a+b)/2}=\frac{b-a}{a+b}$. 
Thus, a 5\% target means that the confidence-sequence interval half-width is at most 5\% of its midpoint. 
For the biodiversity credits application (see \cref{sec:biodiversity-setting}), we specialize $H$ to Shannon entropy $H(\mathbf{p})=-\sum_{s\in\mathcal L}p_s\log p_s$, and consider pooled Hill--Shannon diversity $P_2:=\exp H(\mathbf{p}^\star)$. 
On the CIFAR-10N dataset, we test a different class of nonlinear scores. We define \emph{worst-three precision} as the mean precision of the three least precise submitted classes $g_{\mathrm{prec}}(\mathbf{C}) = \frac{1}{3}\sum_{r\in\mathcal B_3^{\mathrm{prec}}(\mathbf{C})}\frac{C_{rr}}{N_r},$ where \(\mathcal B_3^{\mathrm{prec}}(\mathbf{C})\) indexes the three smallest classwise precision values \(C_{rr}/N_r\). We similarly define \emph{worst-three taxonomic risk} (see \cref{app:cifar-tail-settings}). We additionally evaluate artifact-class micro-F1 on ReImageNet and token-level NER micro-F1 on CleanCoNLL. Their score definitions and settings are given in \cref{app:reimagenet-f1-settings,app:cleanconll-f1-settings}. The prespecified targets are 5\% relative margin of error for Kenya diversity and CIFAR-10N worst-three precision, 25\% for CIFAR-10N worst-three taxonomic risk, and two percentage points for the two F1 scores.

We compare against two score-specific adaptations of recent confidence-sequence methods: PPR with projection \citep{waudby2020confidence} and adapted Manole--Ramdas stitching \citep{manole2023martingale}. Construction details and an uncertainty-only ablation policy are discussed in greater detail in \cref{app:nonlinear-results}. 
Across all five nonlinear settings, ATLAS-PAS has the smallest median audit count to reach the target precision. Relative to ATLAS-SRS, it requires 36--71\% fewer audits. ATLAS-PAS also reduces the median audit count by 6--38\% relative to the adaptive uncertainty-only ablation (\cref{app:nonlinear-results}), showing that downstream-score information adds value beyond label uncertainty alone. PPR with projection reaches the target precision only for ReImageNet F1, while Manole--Ramdas stitching does not attain the target in any of the five settings within 10,000 audits. 
The four main comparisons are displayed in \cref{fig:p2}, and the CIFAR-10N taxonomic-risk comparison is in \cref{fig:cifar-taxonomic-risk}. The numerical results are reported in \cref{app:additional-results}, especially \cref{tab:nonlinear-audits,tab:nonlinear-coverage}.

\begin{figure}[t]
    \centering 
    \includegraphics[width=\linewidth]{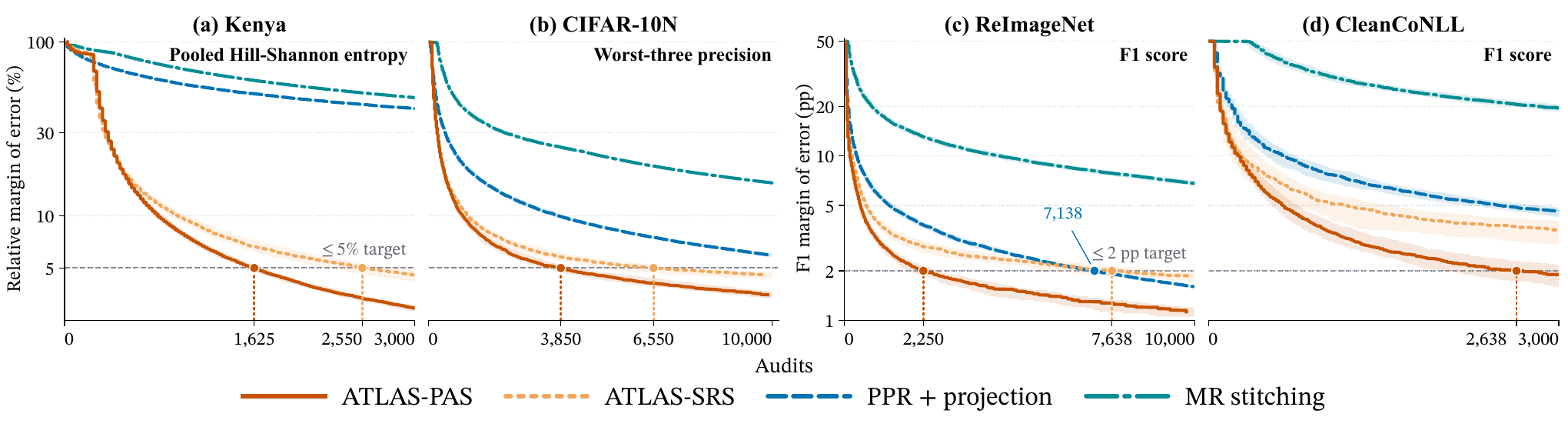}
\caption{
Anytime-valid certification of nonlinear scores on four datasets: \textbf{(A)}~pooled diversity $P_2$ on the Kenya biodiversity dataset, \textbf{(B)}~worst-three precision on CIFAR-10N, and \textbf{(C,D)}~F1 score on ReImageNet and CleanCoNLL.
Curves show median interval half-widths over 100 paired audit paths; shaded regions show interquartile ranges. Panels (A,B) express half-width as a percentage of the interval midpoint; panels (C,D) use percentage points. We compare ATLAS-PAS, ATLAS-SRS, PPR with projection, and adapted Manole--Ramdas stitching. Dashed horizontal lines indicate the precision targets. Full results, including for an uncertainty-only policy ablation, are in \cref{app:nonlinear-results}.
}
    \label{fig:p2}
\end{figure}


%% file: sections/bayesian.tex
\section{A Bayesian, Model-Based Alternative for Efficient Auditing}
\label{sec:bayes}


\Cref{sec:linear,sec:entropy} provide distribution-free guarantees but can require an impractically large number of expert audits. We therefore consider a model-based alternative that shares information across related labels and strata when inferring the composition of unaudited records, trading robustness to model specification for greater statistical efficiency. Importantly, the same task-aware auditing framework accommodates both settings.

We use an approximate hierarchical Dirichlet--multinomial update. For a submitted label $r$ and stratum $g$, let $\bm{\theta}_r$ denote the overall correction distribution and let $\bm{\theta}_g$ denote the stratum-specific correction distribution. We use a hierarchical Dirichlet--multinomial model \citep{kim2012bayesian}: $\bm{\theta}_r\sim\operatorname{Dirichlet}(\kappa_r\{\beta \mathbf{e}_r+(1-\beta)\bar{\mathbf{g}}\}+ \mathbf{A}_{r,t})$, and $\bm{\theta}_{g}\mid \bm{\theta}_{r(g)} \sim\operatorname{Dirichlet}(\kappa_c\,\bm{\theta}_{r(g)}+ \mathbf{A}_{g,t})$,
with the remaining verified-label counts  in stratum~$g$ modeled as $\operatorname{Multinomial}(M_{g,t},\,\bm{\theta}_{g}).$
Here, $\mathbf{A}_{r,t}$ and $\mathbf{A}_{g,t}$ are vectors of verified-label counts audited so far, $M_{g,t}$ is the number of unaudited records remaining in stratum~$g$, $\mathbf{e}_r$ places prior mass on the submitted label $r$ being correct, and $\bar{\mathbf{g}}$ summarizes the smoothed share of verified labels observed across all labels. 
The concentration parameters $\kappa_r$ and $\kappa_c$ control how strongly information is shared across labels and strata. As audits accumulate within a stratum, its local counts increasingly dominate the shared prior (\cref{app:linear}).

We use the same priority-aware sampling policy as in \cref{sec:pas}. At audit~$t$, we select a stratum~$G_t$, sample one unverified record uniformly from that stratum, reveal its verified label~$V_t$, and update the observed correction counts. Beginning after audit 100 and every 25 audits thereafter, we propagate posterior uncertainty to the downstream score by drawing stratum-level correction distributions from the posterior and, conditional on each draw, imputing verified labels for all remaining unaudited records. Each posterior draw defines a complete, corrected dataset consistent with all audits observed so far. We compute $P_2$ on each of 400 datasets, yielding a posterior distribution over the nonlinear downstream score without relying on the linear-envelope construction of \cref{sec:entropy}. We report the 2.25th and 97.75th empirical quantiles as a 95.5\% credible interval and stop once the interval satisfies the prespecified 5\% relative-accuracy target. 

\Cref{tab:bayes} shows substantial empirical savings on the Kenyan dataset: the approximate Bayesian approach with PAS requires a median of only 250 audits compared with 500 under random sampling. Under both PAS and SRS, the credible intervals covered the true diversity
at stopping on 96 of 100 paths, despite having no anytime-valid coverage guarantee. These results illustrate the potential efficiency gains from modeling assumptions and motivate developing less conservative, distribution-free confidence sequences that retain anytime-valid guarantees at practical auditing budgets.

\begin{table}[t]
\caption{On the Kenya biodiversity data, the number of audits required by the approximate Bayesian procedure to estimate the entropy score $P_2$ to 5\% relative accuracy at 95.5\% nominal credibility under the model. Results are over 100 audit paths per sampling strategy; all paths reached the target within 1,000 audits.
Coverage is the fraction of audit paths whose interval contains the true score $\theta^\star$ at stopping time. More details in \cref{app:bayes-settings}.}
\label{tab:bayes}
\centering\small
\setlength{\tabcolsep}{3pt}
\renewcommand{\arraystretch}{1}

\begin{tabular*}{\linewidth}{@{\extracolsep{\fill}}lcccc@{}}
\toprule
\textbf{Sampling strategy}
& \textbf{Median}
& \textbf{Coverage at stopping}
& \textbf{Whole path coverage}\\
\midrule
Priority-aware (Bayesian + PAS)
& {250}\,{\scriptsize [225, 250]}
& 96/100
& 71/100\\
Random (Bayesian + SRS)
& 500\,{\scriptsize [450, 550]}
& 96/100
& 53/100\\
\bottomrule
\end{tabular*}
\end{table}


%% file: sections/related_work.tex
\section{Related Work}
\label{sec:related}

\paragraph{Anytime-valid auditing.} 
Classical concentration bounds such as Serfling's inequality provide finite-sample guarantees for sampling without replacement at a fixed sample size \citep{serfling1974probability,bardenet2015concentration}. More recently, \citet{shekhar2026tighter} sharpened fixed-sample intervals for the population mean.
Confidence sequences, introduced by \citet{darling1967confidence}, remain valid simultaneously over time and permit adaptive stopping \citep{howard2021time,ramdas2023game}. \citet{waudby2020confidence} develop confidence sequences for uniform sampling without replacement, and \citet{shekhar2023risk} construct them for weighted means under adaptive, nonuniform sampling. For label audits, our \algo-PAS uses Bayesian predictions of label corrections to approximate the variance-minimizing allocation for the downstream score. Across mislabel-rate benchmarks, \algo-PAS typically reached the target precision with substantially fewer audits than these baselines (\cref{tab:errrate}). For nonlinear scores such as entropy, one can exploit the convexity of negative entropy \citep{manole2023martingale} or bound the score over a confidence region for the distribution of verified labels in the full dataset \citep{waudby2020confidence,ryu2024gambling}. Such projections can be conservative, as \citet{kaido2019confidence} discuss in partially identified models. We use the structure of each score directly. We certify the Hill--Shannon biodiversity score using a linear-envelope representation of entropy and the F1 score by expressing it as a ratio of two linear scores. In our experiments, our method typically reached the target precision with substantially fewer audits than competing anytime-valid methods (\cref{tab:nonlinear-audits}).


\paragraph{Prediction-assisted inference.}
Verified labels can correct for prediction error in regression \citep{wang2020post} and in prediction-powered inference more broadly \citep{angelopoulos2023prediction}. \citet{kilian2026anytime} use prior information about prediction accuracy to sharpen asymptotic confidence sequences as labeled and unlabeled samples accumulate. \citet{csillag2025prediction} give a general construction for incorporating predictions into e-processes and show that it preserves anytime validity. In our audit, Bayesian predictions of categorical label corrections guide which records to verify in a fixed dataset and, together with the verified labels, yield a finite-sample confidence sequence for the downstream score.

\paragraph{Active evaluation.}
Active evaluation selects test points to label to estimate or compare model performance \citep{sawade2012active,kossen2021active,kossen2022active}. Approaches include sequential planning \citep{mittal2024adaptive}, active statistical inference \citep{zrnic2024active}, and sampling rules robust to inaccurate uncertainty estimates \citep{li2026robust}. \citet{xu2024active} acquire labels for anytime-valid risk control, while \citet{ashouritaklimi2026prediction} use predictions to decide which test point to label next and derive asymptotically valid confidence intervals for the model's average loss. \citet{kay2025consensus} choose which records to label to determine which classifier performs best, and in biodiversity surveys, \citet{haucke2026targeted} select records for expert review to estimate species occupancy. Every record in our fixed dataset already has a submitted label. We choose which labels to verify based on how correcting them could affect the downstream score.

%% file: sections/conclusion.tex
\section{Conclusion}
\label{sec:conclusion}

Our work challenges a common assumption within machine learning communities: that the labels we treat as ground truth are reliable, even when labeling errors can materially distort the resulting  conclusions. When these labels feed into complex downstream score functions, we show that the structure of these score functions can be leveraged to guide priority-aware auditing, substantially reducing the number of expert audits needed to obtain narrow confidence sequences while retaining finite-sample, distribution-free, and anytime-valid guarantees. Our real-world case study highlights the gap between statistical validity and practical feasibility: classical tools such as Serfling's inequality provide rigorous finite-sample guarantees, but may require prohibitively many audits and their fixed-sample guarantees need not hold under data-dependent stopping.

Methodologically, we represent label corrections with a confusion matrix and use Bayesian predictions to guide audits toward corrections that matter most for the downstream score. \algo provides finite-sample guarantees under adaptive sampling and stopping and, in our experiments, typically reached the target precision with fewer audits than competing anytime-valid methods. The efficiency and validity of our approach is relevant to many high-stakes settings, including the original motivation of auditing wildlife observations to determine payments in biodiversity credit programs. 



\paragraph{Limitations and future work.} 
The gains from priority-aware sampling scale with increasing heterogeneity in score-relevant uncertainty across strata, and are therefore limited on highly homogeneous data. Additionally, we treat expert-verified labels as ground truth, although expert verification may itself be subject to error. 
An open direction is to develop direct sequential estimators for nonlinear scores 
that support valid e-processes and further reduce audit costs.
Our model-based Bayesian alternative (\cref{sec:bayes}) provides a useful efficiency benchmark: an open question is how closely distribution-free, anytime-valid procedures can approach this empirical efficiency, and under what conditions. 
Further extensions include heterogeneous audit costs, multiple auditors, and nonstationary confusion patterns.

%% file: appendices/top_of_appendix.tex
\section{Audits for Biodiversity Credit Programs}
\label{sec:biodiversity-setting}

To establish market credibility and prevent greenwashing, voluntary biodiversity credit schemes require standardized monitoring frameworks, robust baseline metrics, and transparent verification processes \citep{wunder2025biodiversity,gonzalez2026data}. A project site is a landowner enrolled in a biodiversity credits program, such as a cattle rancher in rural Kenya. Each project site collects data using sensors placed on their landscape, recording wildlife observations. These data points are labeled with the specific species observed, which determine the biodiversity index for the landscape. 

As a representative application, we consider the Plan Vivo ``multimetric'' \citep{planvivo2026} which includes three species-based pillars\footnote{The multimetric also includes two habitat pillars. Here, we only consider the species-specific pillars.}: species richness $P_1$ (number of unique species), species diversity $P_2$ (relative abundance of species), and taxonomic diversity $P_3$ (dissimilarity in the taxonomic tree). The multimetric determines the number of biodiversity certificates that the landowner is awarded, which then determines their payment.

For a candidate confusion matrix $\mathbf C$, let $m_s(\mathbf C):=\sum_{r\in\mathcal R} C_{rs}$ and $p_s(\mathbf C):=m_s(\mathbf C)/N$ denote the corrected abundance and population proportion of species $s$, respectively, and let $\mathbf p(\mathbf C):=(p_s(\mathbf C))_{s\in\mathcal L}$.

Species are categorized into disjoint groups $\calG=\{G_1,\dots,G_K\}$, where each group $G_k\subseteq\calL$ is a category such as birds, bats, or medium-to-large mammals. The within-group relative abundance of species $s\in G_k$ is
\begin{equation}
p^{(G_k)}_s(\mathbf{C})=\frac{m_s(\mathbf{C})}{\sum_{j\in G_k} m_j(\mathbf{C})},\qquad s\in G_k.
\label{eq:groupcomp}
\end{equation}
Let $z_s(\mathbf{C})=\ind{m_s(\mathbf{C})>0}$ indicate whether species $s$ is present in the corrected labels, $S_k(\mathbf{C})=\sum_{s\in G_k}z_s(\mathbf{C})$ denote the number of species in group $k$, and $S(\mathbf{C})=\sum_{G_k\in\calG}S_k(\mathbf{C})$ denote the total number of present species. The three pillars of the multimetric can then be expressed as:
\begin{align}
P_1(\mathbf C) &= S(\mathbf C), \label{eq:P1}\\
P_2(\mathbf C) &= \Dsm\!\bigl(\mathbf p(\mathbf C)\bigr), \qquad \Dsm(\mathbf p) = \exp\!\left(-\sum_{s\in\calL}p_s\log p_s\right),\label{eq:P2}\\
P_3(\mathbf C) &= \frac{\sum_{i<j}\omega_{ij}z_i(\mathbf C)z_j(\mathbf C)}{\sum_{i<j}z_i(\mathbf C)z_j(\mathbf C)}+\sum_{G_k\in\calG}\frac{S_k(\mathbf C)}{S(\mathbf C)}\,\Delta\big(\mathbf p^{(G_k)}(\mathbf C)\big). \label{eq:P3}
\end{align}
Here, the within-group taxonomic dissimilarity is $\Delta(\mathbf p)=\sum_{i<j}\omega_{ij}p_ip_j/\sum_{i<j}p_ip_j$, where $\omega_{ij}\in[0,\Omega]$ denotes taxonomic distance and $\Omega$ is some maximal taxonomic distance. The multimetric then aggregates the three pillar values, e.g., as a weighted sum: $M(\mathbf{C})=\Phi(P_1(\mathbf C),P_2(\mathbf C),P_3(\mathbf C))$. Notably, this scoring function is both nonlinear ($P_2$ and $P_3$) and nonsmooth ($P_1$ and $P_3$).

%% file: appendices/notation.tex
\section{Notation}
\label{app:notation}

As a convention, we use boldface for vectors (e.g., $\mathbf{p}^\star$) and matrices (e.g., $\mathbf{C}^\star$) and use ordinary type to refer to scalar quantities including individual entries (e.g., $C^\star_{rs}$, $p_s^\star$).

\begin{table}[H]
\centering
\caption{Core notation used in the main text. Symbols used exclusively in the appendix are defined where they are first introduced.}
\label{tab:notation}
\footnotesize
\setlength{\tabcolsep}{4pt}
\renewcommand{\arraystretch}{0.98}
\begin{tabular}{@{}p{0.27\textwidth}p{0.67\textwidth}@{}}
\toprule
\textbf{Symbol} & \textbf{Meaning} \\
\midrule
\addlinespace[0.5em]
\multicolumn{2}{l}{\textbf{Data and labels (\Cref{sec:setting})}} \\
\addlinespace[0.2em]
$N$
& Number of records, indexed by $i\in[N]$. \\

$\mathcal L=\{1,\ldots,L\}$, $\mathcal R\subseteq\mathcal L$
& Full label vocabulary and subset of submitted labels;
$R=|\mathcal R|$. \\

$y_i\in\mathcal R,\ v_i\in\mathcal L$
& Submitted label $y_i$ (known) and expert-verified label $v_i$ (unknown until audit) of record~$i$. \\

$\mathbf{C^\star},\ \mathbf C$
& True confusion matrix and a candidate confusion matrix. Their entries are 
$C^\star_{rs}=\#\{i:y_i=r,\ v_i=s\}$. \\

$N_r$
& Number of records submitted with label $r$:
$N_r=\sum_s C^\star_{rs}$. \\

$g=(r,a),\ r(g)$
& Audit stratum defined by submitted label $r$ and metadata $a$; 
$r(g)$ returns its submitted label. \\

\addlinespace[0.5em]
\multicolumn{2}{l}{\textbf{Audit process (\Cref{sec:linear})}} \\
\addlinespace[0.2em]
$t,\ \mathcal F_{t-1}$
& Audit index and audit history available before audit $t$. \\
$\mathcal A_{t-1}$ & Records audited before audit $t$. \\
$M_{g,t}$
& Number of unaudited records remaining in stratum $g$. \\

$G_t,\ V_t,\ \pi_{g,t}$
& Selected stratum, revealed verified label, and (predictable) sampling probability for stratum $g$ at audit $t$. \\

\addlinespace[0.5em]
\multicolumn{2}{l}{\textbf{Linear-score inference (\Cref{sec:linear})}} \\
\addlinespace[0.2em]
$\mathbf{D},\ \mathbf d_r^\top$
& Linear-score contribution matrix and its $r$th row. \\
$\theta^\star$ & True (unknown) linear score,
$\theta^\star=N^{-1}\sum_i D_{y_i,v_i}$. \\

$\mu_{g,t},\ \sigma^2_{g,t},\ b_{g,t}$
& True mean, variance, and prediction error of the remaining score contributions in stratum
$g$, with
$b_{g,t}=\mu_{g,t}-\mathbf d_{r(g)}^\top\mathbf q_{g,t}$. \\

$\mathbf q_{g,t},\ \widehat\theta_{t-1},\ Z_t$
& Predictive distribution over verified labels for the remaining records in $g$, predicted full-data score, and
conditionally unbiased audit observation. \\

$m,\ \lambda_t^\pm(m),\ E_t^\pm(m)$
& Candidate score value, predictable betting stakes, and the two
e-processes testing null hypothesis $H_0 : \theta^\star = m$. \\

$I_t,\ \alpha,\ \tau$
& Confidence set after $t$ audits (used to form confidence sequence), miscoverage level, and stopping half-width tolerance. \\

\addlinespace[0.5em]
\multicolumn{2}{l}{\textbf{\small Priority-aware sampling (\Cref{sec:pas})}} \\
\addlinespace[0.2em]
$\varepsilon,\ \widehat\sigma_{g,t}$
& SRS mixing weight, and
predicted within-stratum standard deviation. \\

\addlinespace[0.5em]
\multicolumn{2}{l}{\textbf{Nonlinear scores (\Cref{sec:entropy})}} \\
\addlinespace[0.2em]
$\mathbf{p}^\star$ & True corrected label proportions, $p_s^\star=N^{-1}\sum_r C_{rs}^\star$. \\
$H(\mathbf{p}),\ H^\star,\ P_2$
& Concave score on a label composition $\mathbf{p}$, specialized to Shannon entropy; $H^\star=H(\mathbf{p^\star})$; and pooled Hill--Shannon
diversity $P_2= \exp(H^\star)$. \\

$\bar{\mathbf{p}},\ \widehat{\mathbf{p}}_{t-1},\ \bm{\psi}_s(\bar{\mathbf{p}})$
& Reference composition, plug-in composition, and tangent coefficient; for Shannon entropy, $\psi_s(\bar{\mathbf p})=\log(1/\bar p_s)$. \\

$O_{t,s},\ \mathcal P_t,\ h$
& Number of audited records verified as $s$, feasible composition set, and
candidate entropy value. \\

$\alpha_-,\alpha_+,\ L_t,U_t$
& Lower- and upper-tail failure budgets and the resulting entropy
confidence-sequence endpoints. \\

\addlinespace[0.5em]
\multicolumn{2}{l}{\textbf{Bayesian model (\Cref{sec:bayes})}} \\
\addlinespace[0.2em]
$\bm{\theta}_r,\ \bm{\theta}_{g}$
& Row-level and stratum-level correction distributions. \\
$\mathbf A_{r,t},\ \mathbf A_{g,t}$ & 
Verified-label count
vectors for row $r$ and stratum $g$. \\
\bottomrule
\end{tabular}
\end{table}

%% file: appendices/algos_experiments.tex
\section{Algorithm Pseudocode}
\label{app:algos}

\paragraph{Reading \cref{alg:pas}.}
All predictions, priorities, sampling probabilities, and betting stakes are computed before the next verified label is observed; this predictability is what permits adaptive allocation without sacrificing validity. Verified strata are removed and the sampling distribution is renormalized over the remaining records. Prediction quality affects efficiency only: an inaccurate predictor may slow contraction but does not invalidate the confidence sequence.

\paragraph{Reading \cref{alg:entropy}.}
The two entropy endpoints require different certificates. The upper endpoint uses a predictable linear cross-entropy direction, whereas the lower endpoint
must rule out every composition still compatible with the audits through a conservative convex optimization. Both endpoints use the same audit stream, with tail budgets satisfying $\alpha_{+}+\alpha_{-}\leq\alpha$; the reported
interval is valid only when the numerical lower-end optimization returns a
certified bound rather than an unconstrained approximate solution.

\begin{algorithm}[h!]
\caption{\algo Priority-Aware Sampling (PAS) for a linear score}
\label{alg:pas}
\begin{algorithmic}[1]

\Require Strata $g$, contribution matrix $\mathbf D$, predictor $\Call{Predict}$,
miscoverage $\alpha$, SRS weight $\varepsilon$, target half-width $\tau$

\While{confidence-sequence half-width $>\tau$}

\State \textbf{Predict:}
For each nonempty stratum $g$, estimate
$q_{g,t}$ and
\[
\widehat\mu_{g,t}
=
\sum_s q_{g,t}(s)D_{r(g),s},
\qquad
\widehat\sigma_{g,t}^2
=
\sum_s q_{g,t}(s)
\bigl(D_{r(g),s}-\widehat\mu_{g,t}\bigr)^2 .
\]

\State \textbf{Estimate score:}
\[
\widehat\theta_{t-1}
=
\frac{1}{N}
\left[
\sum_{i\in\mathcal A_{t-1}}D_{y_i,v_i}
+
\sum_g M_{g,t}\widehat\mu_{g,t}
\right].
\]

\State \textbf{Update sampling probabilities:}
\[
\pi_{g,t}
=
\varepsilon\frac{M_{g,t}}{\sum_k M_{k,t}}
+
(1-\varepsilon)
\frac{M_{g,t}\widehat\sigma_{g,t}}
{\sum_k M_{k,t}\widehat\sigma_{k,t}} .
\]
\Comment{Mix residual-record SRS with variance-based priority.}

\State \textbf{Set stakes:}
Compute the predictable range of the next $Z_t$ over all possible
strata and verified labels; choose capped stakes
$\lambda_t^\pm(m)$ for each candidate $m$.

\State \textbf{Audit:}
Draw $G_t\sim\bm{\pi}_t$ and uniformly audit one unaudited record in $G_t$,
revealing $V_t$.

\State \textbf{Observe:}
\[
Z_t
=
\widehat\theta_{t-1}
+
\frac{M_{G_t,t}}{N\pi_{G_t,t}}
\left(
D_{r(G_t),V_t}-\widehat\mu_{G_t,t}
\right).
\]

\State \textbf{Update e-processes:}
\[
E_t^+(m)
=
E_{t-1}^+(m)
\bigl[1+\lambda_t^+(m)(Z_t-m)\bigr],
\]
\[
E_t^-(m)
=
E_{t-1}^-(m)
\bigl[1-\lambda_t^-(m)(Z_t-m)\bigr].
\]

\State \textbf{Report:}
\[
\mathcal I_t
=
\left\{
m:
E_s^\pm(m)<2/\alpha
\ \forall s\le t
\right\},
\]
and update the audited set and remaining stratum sizes.

\EndWhile

\State \Return $\mathcal I_t$

\end{algorithmic}
\end{algorithm}

\begin{algorithm}[t]
\caption{\algo Distribution-free audit of entropy}
\label{alg:entropy}
\begin{algorithmic}[1]

\Require Strata $g$, predictor $\Call{Predict}$,
tail budgets $\alpha_\pm$, smoothing $\eta$,
target log-width $\tau=\log(21/19)$

\While{Current endpoint width $U_{t-1}-L_{t-1}>\tau$}

\State \textbf{Predict composition:}
Estimate $\mathbf{q}_{g,t}$ for each nonempty stratum and form
\[
\mathbf{\widehat p_{t-1}} = \frac{1}{N}
\left[\mathbf O_{t-1}+\sum_g M_{g,t}\mathbf{q}_{g,t}\right],
\qquad
\mathbf O_{t-1}=\bigl(O_{t-1,s}\bigr)_{s\in\mathcal L}.\]

\State \textbf{Compute cross-entropy weights:}
\[
\mathbf r_t=(1-\eta)\mathbf{\widehat p}_{t-1}+\frac{\eta}{L}\mathbf{1},
\qquad
\bm{\psi}_t=-\log \mathbf r_t .
\]

\State \textbf{Estimate stratum influence:}
\[
\widehat{\mu}_{g,t}=\mathbf{q}_{g,t}^{\top}\bm{\psi}_t,
\qquad
\widehat\sigma_{g,t}^2
=
\sum_s {q}_{g,t}(s)(\psi_{t,s}-\hat{\mu}_{g,t})^2 .
\]

\State \textbf{Update sampling probabilities:}
\[\bm{\pi}_t=0.4\,\bm{\pi}_t^{\rm prop}+0.3\,\bm{\pi}_t^{\rm floor}+0.3\,\bm{\pi}_t^{\rm pri},\]
where\[\pi_{g,t}^{\rm prop}\propto M_{g,t},
\qquad\pi_{g,t}^{\rm floor}\propto M_{g,t}(1+\hat{\mu}_{g,t}),
\qquad\pi_{g,t}^{\rm pri}\propto M_{g,t}\widehat\sigma_{g,t}.\]

\State \textbf{Set stakes:}
Choose predictable upper/lower stakes from the range of $\bm{\psi}_t$,
the possible Horvitz--Thompson weights, and current endpoints
$L_{t-1},U_{t-1}$.

\State \textbf{Audit:}
Draw $G_t\sim\bm{\pi}_t$ and uniformly audit one unaudited record in $G_t$,
revealing $V_t$.

\State \textbf{Upper endpoint:}
Form the prediction-corrected directional observation $Z_t^U$
using contributions $\bm{\psi}_t$, update
\[
E_t^U(h)
=
E_{t-1}^U(h)
\bigl[1-\lambda_t^U(Z_t^U-h)\bigr],
\]
and set $U_t$ to the smallest $h$ ruled out by the upper e-process.

\State \textbf{Lower endpoint:} For each candidate composition $\mathbf p$, form
$Z_t(\mathbf p)$ using the smoothed composition
$\mathbf{\widetilde p}=(N \mathbf p+\mathbf 1)/(N+L)$ and update
\[
F_t(\mathbf p;h)
=
F_{t-1}(\mathbf p;h)
+
\log\!\left[
1+\lambda_t^L\{Z_t(\mathbf p)-h-c\}
\right],
\qquad
c=\log(1+L/N).
\]

\State \textbf{Profile over feasible compositions:}
For each $h$, compute a certified lower bound on
\[
\min_{p\in\mathcal P_t}F_t(\mathbf p;h)
\]
by projected gradient plus a supporting-plane bound;
use bisection in $h$ to obtain $L_t$.

\State \textbf{Report:}
\[
P_2\in[e^{L_t},e^{U_t}],
\qquad
\widehat P_{2,t}
=
\frac{2e^{L_t+U_t}}{e^{L_t}+e^{U_t}},
\]
then update counts and remaining records.

\EndWhile

\State \Return $[e^{L_t},e^{U_t}]$ and $\widehat P_{2,t}$

\end{algorithmic}
\end{algorithm}

\clearpage
\section{Experiment Settings and Additional Results}
\label{app:additional-results}
All error-rate experiments use nominal confidence \(95.5\%\) (\(\alpha=0.045\)). PAS and SRS use the same prediction-corrected betting confidence sequence and differ only in their predictable sampling probabilities. Records are sampled uniformly within the selected stratum. Paired policies use identical randomized within-stratum record orders.

The mislabel-rate experiments use 100 paired paths with seeds \(2026110900,\ldots,2026110999\). Predictors start from a common, flat correctness probability of \(0.5\) and are updated using audited labels only. Sampling decisions are made before the selected reference label is revealed. Intervals are evaluated after every audit, and target attainment is recorded from audit 100 onward. Runs continue through 5,000 audits even after reaching the target, allowing us to evaluate subsequent interval widths and whole-path coverage.

For each precision level, we record the first eligible update at which a nonempty interval attains that precision. The real-data audit curves plot the median of these per-path attainment times against precision and shaded regions show their interquartile ranges. Paths that do not attain a target within budget remain censored and are included when computing quantiles. Empty intervals do not count as precision attainment, and paths with coverage failures are not discarded.

ATLAS allocation mixtures are fixed within each experimental setting. Mislabel-rate and F1 experiments use a $0.30$ SRS component.
The Kenya diversity and CIFAR nonlinear experiments use proportional
components $0.40$ and $0.70$, respectively, together with score-dependent range and variance components. The Bayesian experiment uses a $0.20$ proportional component.

\begin{table}[h]
\centering
\caption{Frozen ATLAS settings for the four mislabel-rate experiments.
Comparator-specific settings are given in \cref{app:linear-comparators}.}
\label{tab:linear-settings}
\footnotesize
\setlength{\tabcolsep}{4pt}
\begin{tabular}{@{}lrrrr@{}}
\toprule
\textbf{Setting} & \textbf{Kenya} & \textbf{CIFAR-10N}
& \textbf{ReImageNet} & \textbf{CleanCoNLL} \\
\midrule
Population size & 33,403 & 50,000 & 30,908 & 301,353 \\
Audit strata & 46 & 40 & 964 & 9 \\
Paths per policy & \multicolumn{4}{c}{100} \\
Interval update spacing & \multicolumn{4}{c}{1 audit} \\
Earliest target attainment & \multicolumn{4}{c}{100 audits} \\
Maximum audits & \multicolumn{4}{c}{5,000} \\
Target half-width & \multicolumn{4}{c}{0.02} \\
PAS SRS weight $\varepsilon$ & \multicolumn{4}{c}{0.30} \\
Initial correctness probability & \multicolumn{4}{c}{0.50} \\
Row/secondary shrinkage & \multicolumn{4}{c}{$2/1$} \\
Stake effective horizon & \multicolumn{4}{c}{800} \\
Stake cap fraction & \multicolumn{4}{c}{0.90} \\
Candidate-grid spacing & \multicolumn{4}{c}{$5\times10^{-4}$} \\
\bottomrule
\end{tabular}
\end{table}

\subsection{Kenya Biodiversity Setting}
\subsubsection{Data}
Camera-trap records supplied by our NGO partner, pair submitted identifications with expert corrections. We merge label variants into broad categories (e.g., zebra subspecies into Zebra), retaining all 33,403 records and a shared vocabulary of 50 categories.
\subsubsection{Linear Score: Mislabel Rate}
\label{app:ns-linear-settings}

Submitted labels define 46 audit strata. The predictive error model
starts from error probability \(0.5\) in every stratum, with row and
secondary shrinkage strengths \(2\) and \(1\). PAS combines \(30\%\)
remaining-record SRS with \(70\%\) predicted Neyman allocation.
PAS and SRS share the same initialization, prediction updates, and
confidence-sequence construction. Full results are presented in \cref{tab:errrate} and displayed in \cref{fig:kenyan-mislabel-rate}.

\subsubsection{Nonlinear Score: Pillar 2 Diversity}
\label{app:ns-diversity-settings}

As in the main text, let
\[
m_s(\mathbf C)=\sum_r C_{rs},
\qquad
p_s(\mathbf C)=\frac{m_s(\mathbf C)}{N}
\]
denote the corrected abundance and population proportion of species \(s\)
induced by a candidate confusion matrix \(\mathbf C\). Pillar 2 is the Hill--Shannon
diversity of this corrected composition:
\[
H\!\left(\mathbf{p(C)}\right)
=
-\sum_{s\in\mathcal L}p_s(\mathbf C)\log p_s(\mathbf C),
\qquad
P_2(\mathbf C)
=
\exp\!\left\{H\!\left(\mathbf{p(C)}\right)\right\}.
\]
Thus, \(P_2\) is the effective number of species implied by the corrected
labels: it increases when the composition contains more species or distributes
mass more evenly among them. Label corrections affect \(P_2\) through the
column totals of \(\mathbf C\), and its nonlinearity makes uncertainty in rare or
poorly audited species especially consequential.

The distribution-free \(P_2\) experiment uses 300 nonempty strata obtained by
crossing submitted label with camera station. The prediction hierarchy shrinks
station-level distributions toward their submitted-label row with strength
\(1\), and row distributions toward a pooled composition with strength \(2\).
The row-level predictive base is
\(0.8e_r+0.2\bar p_{t-1}\), where \(e_r\) concentrates on the
submitted label and \(\bar p_{t-1}\) is the smoothed composition
estimated from previously audited labels. This composition is uniform
before the first audit. Row and stratum shrinkage strengths are
\(2\) and \(1\), respectively. The predictable reference composition used in the entropy direction is
\[
\mathbf r_t
=
0.99\,\mathbf {\widehat p}_{t-1}
+
\frac{0.01}{L}\mathbf 1.
\]

The experimental score is evaluated over the full cleaned annotation
vocabulary, including the non-species categories retained in the audit
population.

PAS uses the frozen mixture
\[
\bm{\pi}^{\mathrm{PAS}}_t
=
0.4\,\bm{\pi}^{\mathrm{prop}}_t
+
0.3\,\bm{\pi}^{\mathrm{range}}_t
+
0.3\,\bm{\pi}^{\mathrm{Neyman}}_t.
\]
The proportional component samples according to remaining stratum size, the
range component controls the possible range of the importance-weighted
observation, and the Neyman component uses predicted within-stratum variation
in the current entropy direction.




\begin{table}[h!]
\centering
\caption{Frozen settings for the Kenyan biodiversity data distribution-free diversity
experiment.}
\label{tab:ns-diversity-settings}
\small
\setlength{\tabcolsep}{5pt}
\begin{tabular}{@{}lr@{}}
\toprule
\textbf{Setting} & \textbf{Value} \\
\midrule
Audit strata & 300 \\
Paired paths & 100 \\
Seeds & \(95001,\ldots,95100\) \\

Maximum audits: ATLAS and uncertainty-only & \(3{,}000\) \\
Maximum audits: PPR and MR & \(10{,}000\) \\
Endpoint update spacing & 25 \\
Reported checkpoint spacing & 25 \\
Lower-tail error budget & \(0.0225\) \\
Upper-tail error budget & \(0.0225\) \\
Proportional/range/Neyman weights & \(0.40/0.30/0.30\) \\
Reference smoothing & \(0.01\) \\
Row/stratum shrinkage & \(2/1\) \\
Stake cap fraction & \(0.90\) \\
Entropy resolution & \(3\times10^{-4}\) \\
Profile iteration limit & 60 \\
Target entropy width & \(\log(21/19)\) \\
\bottomrule
\end{tabular}
\end{table}

The entropy interval is transformed outward to the \(P_2\) scale. The stopping criterion \((e^{U_t}-e^{L_t})/(e^{U_t}+e^{L_t})\leq 0.05\)
requires the confidence-sequence half-width to be at most \(5\%\) of its midpoint. SRS and
PAS use the same lower convex-profile certificate, upper affine certificate,
prediction hierarchy, and stake rules.

\paragraph{Results.}The median audits to reach 5\% relative half-width are 1,625 for
ATLAS-PAS, 1,913 for uncertainty-only allocation, and 2,550 for
ATLAS-SRS. All 100 PAS paths and 87 of 100 SRS paths reach the
target within 3,000 audits. Neither PPR projection nor MR stitching
reaches the target within 10,000 audits on any path. Full results are reported in \cref{tab:nonlinear-audits} and displayed in \cref{fig:p2} (a).

\subsubsection{Bayesian Settings}
\label{app:bayes-settings}

\paragraph{Setting.}We evaluate the model-based procedure of \cref{sec:bayes} on the Kenyan
biodiversity dataset. PAS and SRS use the same hierarchical
Dirichlet--multinomial model, posterior computation, credible interval, and
stopping rule; they differ only in audit allocation. Records are sampled
uniformly within the selected stratum, and paired PAS and SRS paths use the
same randomized within-stratum record orders.

The computation uses an empirical-Bayes-style hierarchical
approximation. At each update, the pooled label composition is
re-estimated from audited counts with add-one smoothing and treated
as fixed in the subsequent Dirichlet updates. Its estimation
uncertainty is not integrated through a full joint posterior.
Completed-population draws preserve all audited labels and impute
only the remaining records. The resulting intervals have nominal
credibility under this approximation and carry no distribution-free
anytime-valid coverage guarantee.

The model hyperparameters are fixed at
$\beta=0.5$, $\kappa_r=2$, and $\kappa_c=0.5$, with add-one smoothing of the
pooled verified-label composition. Posterior intervals are evaluated every
25 audits beginning at audit 100. Each evaluation uses 400 completed-population draws. We take the 2.25th and 97.75th percentiles of their entropy values and exponentiate these endpoints to obtain the interval for \(P_2\).

PAS uses the frozen allocation
\[
\bm{\pi}^{\mathrm{PAS}}_t
=
0.2\,\bm{\pi}^{\mathrm{prop}}_t
+
0.8\,\bm{\pi}^{\mathrm{Neyman}}_t,
\qquad
\bm{\pi}^{\mathrm{Neyman}}_{g,t}
\propto
M_{g,t}\hat{\sigma}_{g,t},
\]
where $\hat{\sigma}_{g,t}$ is the predicted within-stratum standard
deviation of the entropy influence direction. Its variance includes the
light-audit adjustment
$1+1/(n_{g,t}+\kappa_c)$, where $n_{g,t}$ denotes the number of records audited in stratum $g$ before audit $t$. Consequently, the adjustment inflates predicted uncertainty only in
lightly audited strata. SRS samples uniformly from all unaudited records.
Neither policy uses a pilot phase.

\paragraph{Coverage.}Coverage at stopping records whether the interval contains the true
score at the first checkpoint satisfying the precision target.
Whole-path coverage requires inclusion at every evaluated checkpoint,
starting at audit 100 and ending at that path's stopping time.
Runs terminate when the target is reached. Coverage at stopping is \(96/100\) for both policies, whereas whole-path coverage is \(71/100\) for PAS and \(53/100\) for SRS. Thus, favorable coverage at the selected stopping time coexists with substantially lower simultaneous coverage along the audit.

\begin{table}[H]
\centering
\caption{Frozen settings for the Bayesian diversity experiment.}
\label{tab:bayes-settings}
\footnotesize
\setlength{\tabcolsep}{5pt}
\begin{tabular}{@{}lr@{}}
\toprule
\textbf{Setting} & \textbf{Value} \\
\midrule
Audit strata & 300 \\
Audit paths per policy & 100 \\
Seeds & $8000,\ldots,8099$ \\
Maximum audits & 1,000 \\
Evaluation schedule & $100,125,\ldots,1000$ \\
Posterior draws per evaluation & 400 \\
Credible-interval percentiles & 2.25 / 97.75 \\
Prior weight $\beta$ & 0.5 \\
Concentrations $\kappa_r/\kappa_c$ & $2/0.5$ \\
Pooled-composition smoothing & Add one \\
PAS proportional/Neyman weights & $0.20/0.80$ \\
Pilot audits & 0 \\
Stopping threshold & \(U_t^H-L_t^H\leq\log(21/19)\) \\
\bottomrule
\end{tabular}
\end{table}

\paragraph{Results.}Median stopping counts are 250 [225, 250] for Bayesian PAS and
500 [450, 550] for Bayesian SRS, where brackets give the
25th and 75th percentiles. Coverage at stopping is 96/100 for
both policies, while coverage at every checked update through
stopping is 71/100 and 53/100, respectively.

\subsection{CIFAR-10N Setting}
\subsubsection{Data}
CIFAR-10N adds noisy human annotations to the ten-class CIFAR-10 image dataset \citep{wei2022cifarn}. We retain all 50,000 training images, pairing \texttt{random\_label1} as the submitted label with \texttt{clean\_label} as the reference; no records are filtered. The 40 strata combine the ten submitted classes with four agreement
patterns: both auxiliary annotations support the submitted label,
exactly one supports it, both agree on another label, or all three
annotations differ.
\subsubsection{Linear Score: Mislabel Rate}
\label{app:cifar-linear-settings}

The additional noisy annotations, \texttt{random\_label2} and
\texttt{random\_label3}, are assumed available before auditing and
define agreement-based strata together with the submitted label.
All strata start with the same predicted error probability of $0.5$;
predictions are subsequently updated using only audited records.
No audit-reference labels are used for pre-audit calibration or
configuration selection.

The experiment uses the frozen settings in \cref{tab:linear-settings}.
Audit-reference labels are revealed only when their records are audited
and are otherwise used solely for evaluating the completed experiments.
Reported audit counts exclude the collection of the pre-existing
auxiliary annotations. Full results are presented in \cref{tab:errrate}, and \cref{fig:cifar-mislabel-rate} shows the mislabel rate certification curves across ATLAS and baselines.

\begin{figure}[h!]
\label{app:cifar-mislabel-figure}
    \centering
    \includegraphics[width=\linewidth]
        {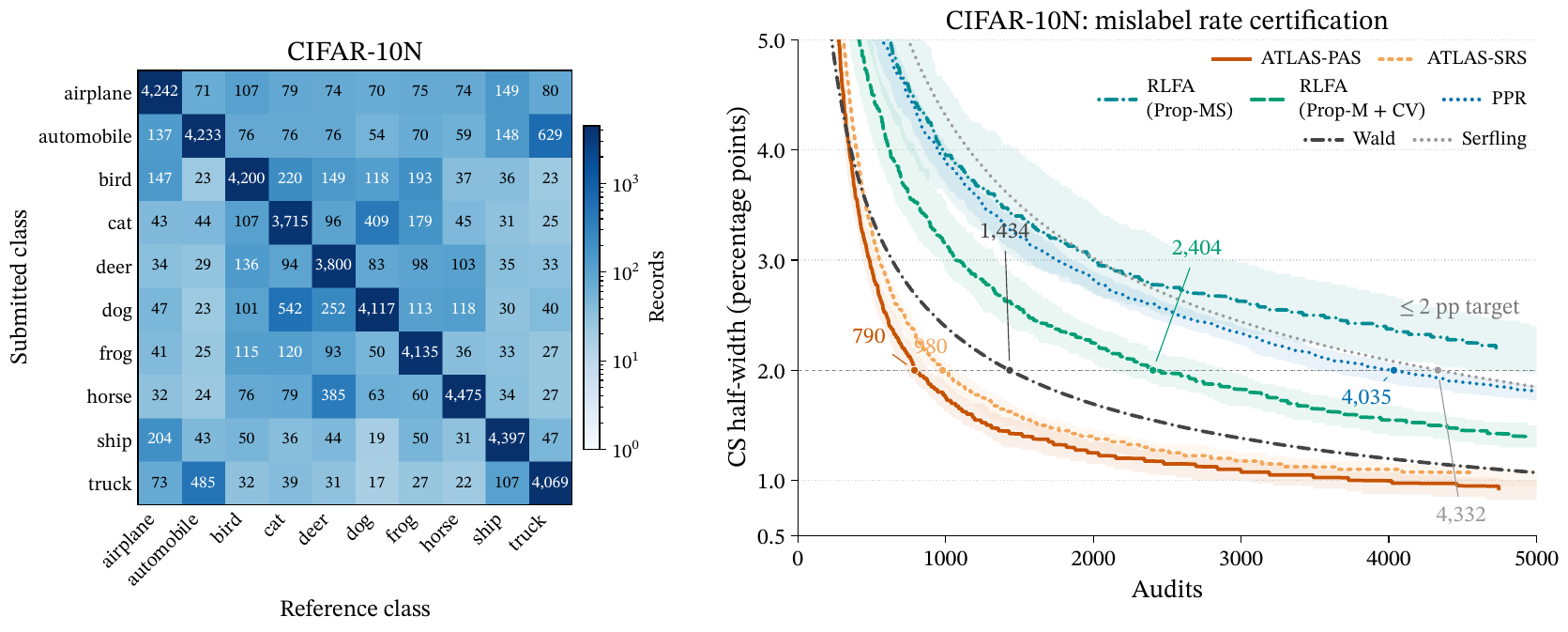}
    \caption{
Audits on CIFAR-10N.
\textbf{(Left)} Confusion matrix. Diagonal entries indicate agreement with the reference labels, while off-diagonal entries indicate label disagreements. Cell annotations give record counts, and color displays them on a log scale.
\textbf{(Right)} Margin of error (half-width of the 95.5\% mislabel-rate confidence sequence) as the number of audits increases with 25th--75th percentile shaded bands. Methods and plotting conventions match \cref{fig:kenyan-mislabel-rate}; dots mark the audits required to reach a half-width of two percentage points. Numerical results are reported in \cref{tab:errrate}.
}
    \label{fig:cifar-mislabel-rate}
\end{figure}

\subsubsection{Nonlinear Score: Worst-Three Scores}
\label{app:cifar-tail-settings}

For a confusion matrix \(\mathbf C\), let
\(N_r=\sum_{s=1}^{10}C_{rs}\) denote the number of records submitted as
class \(r\). The classwise precision is \(p_r(\mathbf C)=C_{rr}/N_r\). We define
worst-three precision as the mean precision of the three least precise
submitted classes:
\[
g_{\mathrm{prec}}(\mathbf C)
=
\min_{\substack{S\subseteq[10]\\ |S|=3}}
\frac{1}{3}\sum_{r\in S}\frac{C_{rr}}{N_r}.
\]

To define taxonomic risk, we map each CIFAR-10 class to a fixed noun synset
in Princeton WordNet 3.0 \citep{fellbaum1998wordnet}. For each pair of classes
\(r,s\), we compute the unweighted shortest-path distance between their
synsets in the undirected WordNet hypernym graph and normalize it by the
largest distance among the ten classes. This gives a fixed semantic-distance
matrix \( \mathbf{D^{\text{tax}}}=(d_{rs})\), with \(d_{rr}=0\) and \(d_{rs}\in[0,1]\).

The classwise taxonomic risk is the average semantic severity of corrections
from submitted class \(r\),
\[
\tau_r(\mathbf C)
=
\frac{1}{N_r}\sum_{s=1}^{10} C_{rs}d_{rs}.
\]
We define worst-three taxonomic risk as the mean risk of the three
highest-risk submitted classes:
\[
g_{\mathrm{tax}}(\mathbf C)
=
\max_{\substack{S\subseteq[10]\\ |S|=3}}
\frac{1}{3}\sum_{r\in S}\tau_r(\mathbf C)
=
\max_{\substack{S\subseteq[10]\\ |S|=3}}
\frac{1}{3}\sum_{r\in S}
\frac{1}{N_r}\sum_{s=1}^{10}C_{rs}d_{rs}.
\]

\paragraph{Confidence-sequence construction.}
Both scores are extrema over the $\binom{10}{3}=120$ triples of
submitted classes. For each fixed triple $S$, its average precision
$\theta_S$ is linear in the confusion matrix because the submitted
class counts $N_r$ are known. Lemma~1 therefore gives an audit
observation with conditional mean $\theta_S$.

Consider worst-three precision, $g=\min_S\theta_S$. To establish that $g>m$, we need evidence
that \emph{every} triple has precision above $m$. We test
$\theta_S\le m$ for each triple and reject $g\le m$ once all 120 tests
have crossed $2/\alpha$, possibly at different audit times. If
$g\le m$, at least one fixed minimizing triple satisfies
$\theta_S\le m$. Rejecting the overall null requires that triple's
test to cross, which has probability at most $\alpha/2$ by Ville's
inequality.

To establish that $g<m$, we can focus on the classes predicted to
have the lowest precision. Before each audit, we select the three
such classes and use their audit observation to update a single
downward betting process. Under $g\ge m$, every triple has precision
at least $m$. Thus the selected observation has conditional mean at
least $m$, even when the selected triple changes between audits.
The betting process is a nonnegative supermartingale, so its
probability of ever crossing $2/\alpha$ is at most $\alpha/2$.
Combining the two tests gives a $1-\alpha$ confidence sequence.

For worst-three taxonomic risk, the target is a maximum. We use all
120 triples for the upper bound and the three classes with the
highest predicted risk for the lower bound. The same proof applies
with the inequalities reversed.

\paragraph{Numerical endpoints.}
We evaluate candidate scores on a grid with spacing $0.001$. We take
the interval from the smallest to the largest surviving grid point,
expand it by one grid step at each end, and clip it to $[0,1]$.
If the grid is empty, we retain the last nonempty interval, starting
from $[0,1]$. This outward rounding preserves coverage: the upward
test rejects a lower tail of grid values, and the downward test
rejects an upper tail. Except on the two crossing events bounded
above, the grid ceiling of the true score survives the upward test
and its grid floor survives the downward test. The expanded interval
therefore contains the true score.

\paragraph{Predictor and sampling policy.}
Before incorporating a class's own audit results, the predictor gives
\(80\%\) weight to its submitted label and \(20\%\) to the label
frequencies estimated from all audits so far. Before the first audit,
these frequencies are initialized uniformly across the ten classes.
The row and secondary shrinkage strengths are \(2\) and \(1\),
respectively.

Because the worst three classes are unknown, PAS assigns \(70\%\) of
its sampling probability to simple random sampling among the remaining
records. This ensures that classes outside the predicted worst three
continue to receive audits. The remaining probability is split equally
between the range-based and Neyman allocations. Both prioritize classes
predicted to be among the worst three and classes belonging to triples
that still limit the confidence bound, giving equal weight to the two
criteria. Tied classes share predicted membership weight equally, and
the triple weights are smoothed using a softmax with temperature \(1\).

Development and confirmation use separate randomized audit paths on
the same fixed population. On each confirmation path, we record the
first checkpoint at which the relative half-width $(U_t-L_t)/(U_t+L_t)$
is at most $5\%$ for precision or $25\%$ for taxonomic risk.
\Cref{tab:cifar-tail-settings} gives the full experimental settings.

\begin{table}[h]
\centering
\caption{Frozen settings for the CIFAR-10N worst-three experiments.}
\label{tab:cifar-tail-settings}
\small
\setlength{\tabcolsep}{5pt}
\begin{tabular}{@{}lr@{}}
\toprule
\textbf{Setting} & \textbf{Value} \\
\midrule
Submitted rows & 10 \\
Tail size & 3 \\
Fixed three-row subsets & 120 \\
Scores & worst-three precision; worst-three taxonomic risk \\
Paired confirmation paths & 100 \\
Confirmation seeds & \(2026100300,\ldots,2026100399\) \\
Maximum audits & \(10{,}000\) \\
Stored checkpoints & \(25,50,75,\ldots,10{,}000\) \\
Confidence level & \(95.5\%\) \\
Candidate-grid spacing & \(0.001\) \\
Stake effective horizon & 800 \\
Stake cap fraction & \(0.90\) \\
Proportional/range/Neyman weights & \(0.70/0.15/0.15\) \\
Active-tail/profile relevance weights & \(0.50/0.50\) \\
Profile temperature & \(1.0\) \\
\bottomrule
\end{tabular}
\end{table}

\paragraph{Results.}
For worst-three precision, the median audit counts were 3,850 for
ATLAS-PAS and 6,550 for ATLAS-SRS. For worst-three taxonomic risk,
the corresponding counts were 3,450 and 5,588. Neither PPR projection
nor MR stitching reached either target on any path within 10,000
audits. \Cref{tab:nonlinear-audits} gives the full comparison.
The precision and taxonomic-risk curves appear in \cref{fig:p2}(b)
and \cref{fig:cifar-taxonomic-risk}, respectively.

\begin{figure}[ht]
\centering
\includegraphics[width=.75\linewidth]{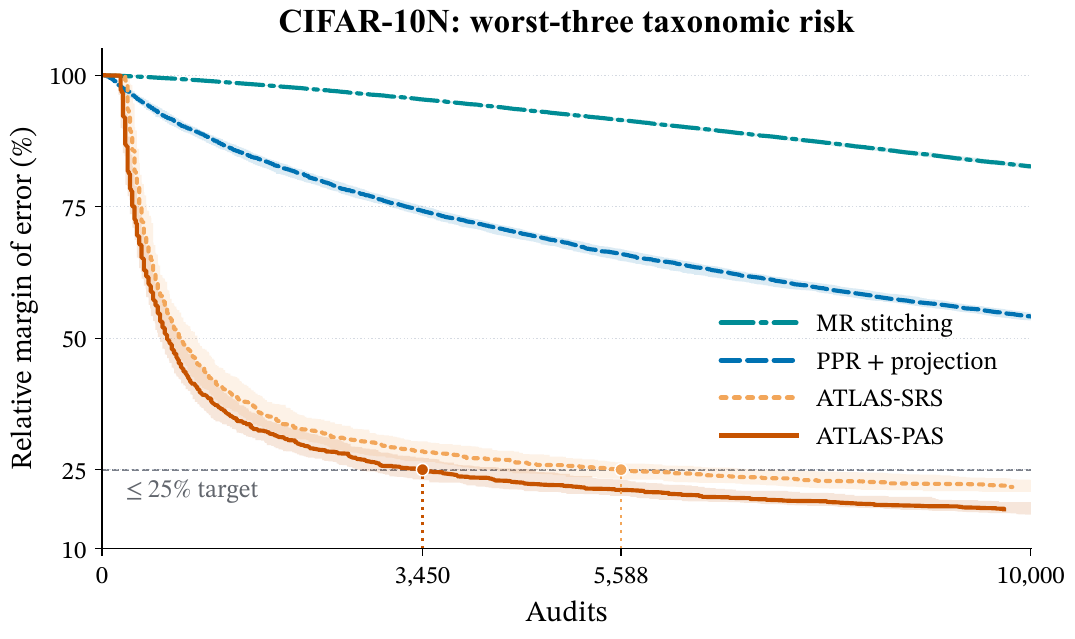}
\caption{Certification of worst-three taxonomic risk on CIFAR-10N at
95.5\% nominal confidence. Curves show median audits to attain each relative half-width over 100 paired paths; shaded regions show interquartile ranges. Dotted guides mark median stopping times
at the 25\% target for ATLAS-PAS and ATLAS-SRS.}
\label{fig:cifar-taxonomic-risk}
\end{figure}




\subsection{ReImageNet Setting}

\subsubsection{Data}
ReImageNet provides revised annotations for the ImageNet-1K validation set \citep{volkov2026imagenet}. We retain 30,908 records having exactly one valid original label and one valid corrected label, excluding multilabel and unresolved cases. In the nonlinear score setting, we compute micro-F1 over 351 artifact classes (human-made objects such as tools, clothing, and furniture). However, we audit all 30,908 records across 964 nonempty submitted-label classes, including those originally labeled outside this group, so the score captures both non-artifacts incorrectly labeled as artifacts and artifacts incorrectly labeled as non-artifacts.
\subsubsection{Linear Score: Mislabel Rate}
\label{app:reimagenet-linear-settings}
We audit the binary disagreement between the original and corrected
labels over all 30,908 retained records. The population mislabel rate
is 8.09\%, and the 964 nonempty submitted-label classes define the
audit strata. All methods use the shared protocol in
\cref{tab:linear-settings}. Median stopping counts are 1,394 for
ATLAS-PAS and 1,539 for ATLAS-SRS; the additional comparator results
are reported in \cref{tab:errrate}. \Cref{fig:reimagenet-mislabel-rate} shows the ReImageNet mislabel rate certification curves across ATLAS and baselines.
\begin{figure}[!h]
\label{app:reimagenet-mislabel-figure}
    \centering
    \includegraphics[width=\linewidth]
        {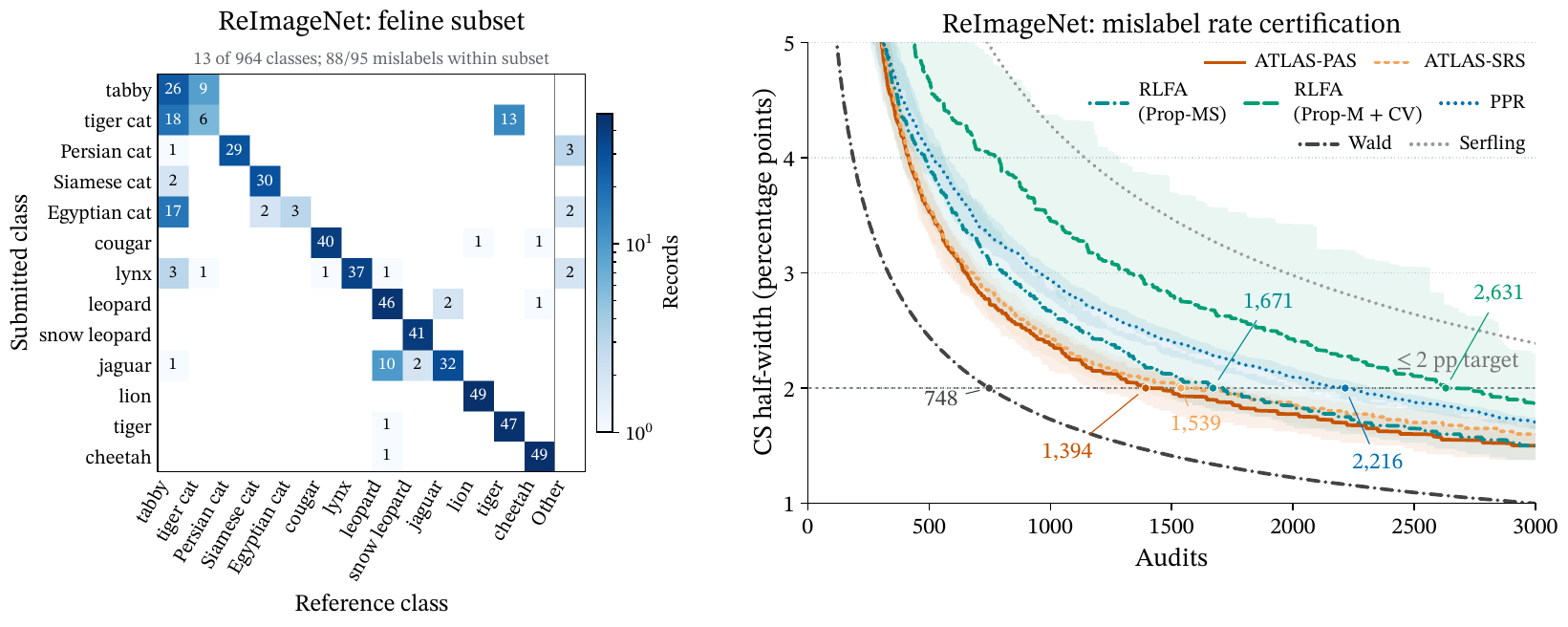}
    \caption{
Audits on ReImageNet.
\textbf{(Left)} Confusion matrix for 13 feline classes, selected for
visualization from the 964 nonempty submitted-label classes.
The Other column collects corrections outside the displayed subset.
Cell annotations give record counts, shown on a logarithmic color
scale. This display restriction does not change the audit population.
\textbf{(Right)} Margin of error (half-width of the 95.5\% mislabel-rate confidence sequence) as the number of audits increases with 25th--75th percentile shaded bands. Methods and plotting conventions match \cref{fig:kenyan-mislabel-rate}; dots mark the audits required to reach a half-width of two percentage points. Numerical results are reported in \cref{tab:errrate}.
}
    \label{fig:reimagenet-mislabel-rate}
\end{figure}

\subsubsection{Nonlinear Score: Artifact-Class Micro-F1}
\label{app:reimagenet-f1-settings}

The target is micro-F1 over a fixed set $\mathcal A$ of 351 artifact
classes, including tools, clothing, and furniture. For submitted-to-
reference confusion counts $C_{rs}$, we define
\[
F_{1,\mathcal A}(\mathbf C)
=
\frac{2\sum_{k\in\mathcal A}C_{kk}}
{\sum_{r\in\mathcal A}\sum_s C_{rs}
+\sum_r\sum_{s\in\mathcal A}C_{rs}}.
\]
A true positive requires an exact match of the original and corrected
fine-grained class. Confusing two different artifact classes therefore
contributes both a false positive and a false negative. Records crossing
the artifact/non-artifact boundary also contribute to the score.

The restriction applies to the scored classes; all 30,908 records
remain eligible for auditing. Sampling uses the same 964 submitted-label
strata as the mislabel-rate experiment. The target half-width is
2 percentage points at 95.5\% nominal confidence, with the shared F1
settings in \cref{app:nonlinear-results}.

\paragraph{Results.}Median stopping counts are 2,250 for ATLAS-PAS, 3,363 for
uncertainty-only allocation, 7,138 for PPR projection, and 7,638
for ATLAS-SRS. MR stitching does not reach the target within
10,000 audits on any path. Full results are reported in \cref{tab:nonlinear-audits} and displayed in \cref{fig:p2} (c).

\subsection{CleanCoNLL Setting}
\subsubsection{Data}
CleanCoNLL corrects named-entity annotations in English CoNLL-2003 news text \citep{rucker2023cleanconll}. We pool all splits, standardize labels to BIO2, and retain 301,353 exactly aligned token pairs, excluding edited or unalignable tokens. The \texttt{O} class remains in the audit population but is excluded from the token-level F1 target.
\subsubsection{Linear Score: Mislabel Rate}
\label{app:cleanconll-linear-settings}
The mislabel-rate experiment includes all 301,353 retained tokens,
including tokens labeled \texttt{O}. Any disagreement between the
submitted and corrected BIO2 tags counts as an error. The nine
submitted tags define the audit strata, and the population mislabel
rate is 1.20\%. We use the shared settings in
\cref{tab:linear-settings}.

Median stopping counts are 281 for ATLAS-PAS and 293 for ATLAS-SRS.
Prop-M + CV has a smaller median of 233, but reaches the target on
75/100 paths within 5,000 audits, compared with 100/100 for ATLAS-PAS. \Cref{tab:errrate} contains full numerical results and \cref{fig:cleanconll-mislabel-rate} shows the CleanCoNLL mislabel rate certification curves across ATLAS and baselines.
\begin{figure}[h!]
\label{app:cleanconll-mislabel-figure}
    \centering
    \includegraphics[width=\linewidth]
        {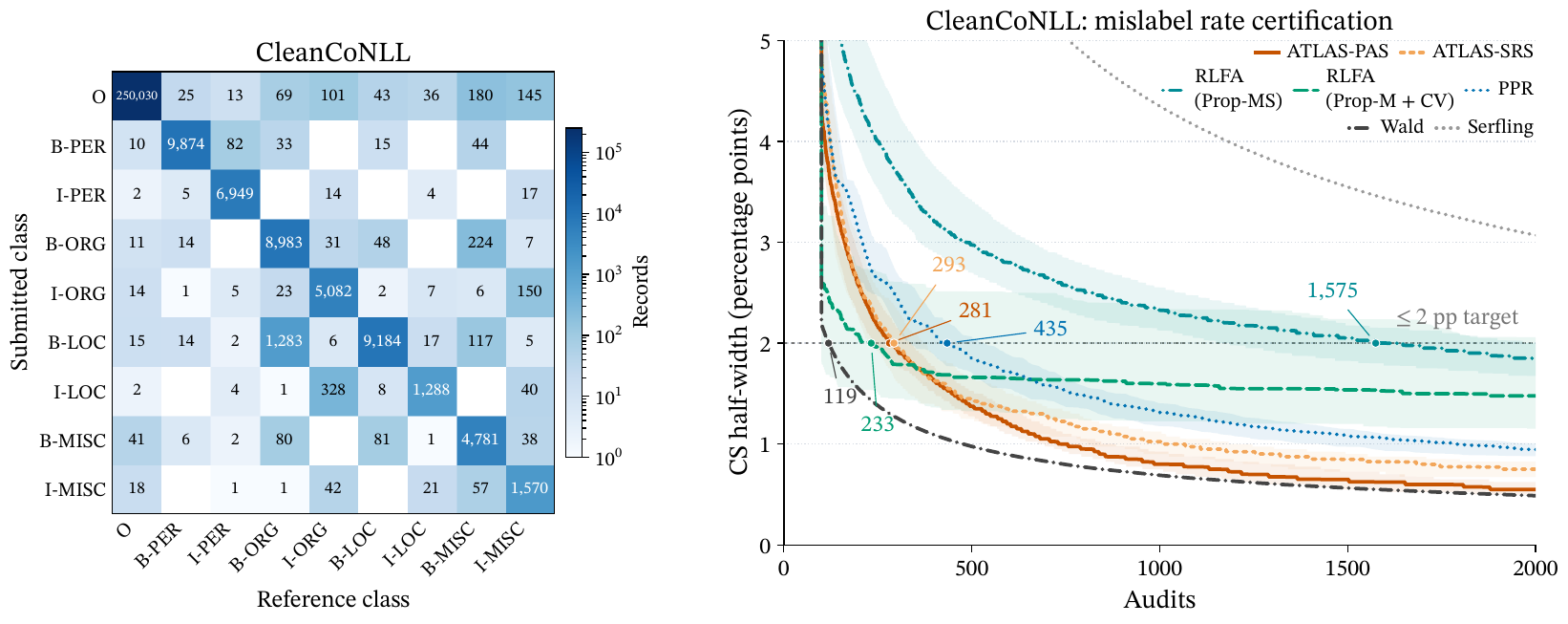}
    \caption{
Audits on CleanCoNLL.
\textbf{(Left)} NER Confusion matrix of tokens. Diagonal entries indicate agreement with the reference labels, while off-diagonal entries indicate label disagreements. Cell annotations give record counts, and color displays them on a log scale.
\textbf{(Right)} Margin of error (half-width of the 95.5\% mislabel-rate confidence sequence) as the number of audits increases with 25th--75th percentile shaded bands. Methods and plotting conventions match \cref{fig:kenyan-mislabel-rate}; dots mark the audits required to reach a half-width of two percentage points. Numerical results are reported in \cref{tab:errrate}.
}
    \label{fig:cleanconll-mislabel-rate}
\end{figure}

\subsubsection{Nonlinear Score: Token-Level NER Micro-F1}
\label{app:cleanconll-f1-settings}

We compute token-level micro-F1 over the eight named-entity BIO2 tags:
the B- and I- forms of PER, ORG, LOC, and MISC. The score uses the
definition in \cref{app:reimagenet-f1-settings}, with $\mathcal A$
equal to these eight tags. A true positive requires an exact tag
match, including both entity type and B/I status. This is a token-level
score rather than an entity-span score.

The \texttt{O} tag denotes the absence of a named entity and is
therefore excluded from the scored classes. All \texttt{O} tokens
remain in the audit population: corrections from \texttt{O} to an
entity tag contribute false negatives, and corrections from an entity
tag to \texttt{O} contribute false positives. Correctly labeled
background tokens do not contribute true positives.

\paragraph{Results.}We retain the full population and its nine submitted-tag strata,
using the shared F1 settings in \cref{app:nonlinear-results}.
The target is an absolute half-width of 2 percentage points.
Median stopping counts are 2,638 for ATLAS-PAS and 2,800 for
uncertainty-only allocation. The medians for ATLAS-SRS, PPR projection,
and MR stitching exceed the 10,000-audit budget. Full results are reported in \cref{tab:nonlinear-audits} and displayed in \cref{fig:p2} (d).

\subsection{Nonlinear Scores: Supplemental Details and Results}
\label{app:nonlinear-results}

\paragraph{Shared F1 settings.}
Both F1 experiments use 100 paired paths with seeds
$96201,\ldots,96300$, a 10,000-audit budget, and precision evaluation
every 25 audits. The confidence level is 95.5\%, and the target
absolute half-width is $0.02$.

The ATLAS predictive base initially mixes $0.8$ mass on the submitted
label with $0.2$ mass on the uniform distribution over the full label
vocabulary. The pooled distribution is subsequently updated from
audited labels, with row and secondary shrinkage strengths $2$ and $1$.
PAS mixes $0.30$ proportional-to-remaining-size sampling with $0.70$
score-aware Neyman allocation. The candidate-score grid spacing is
$0.0005$, the stake effective horizon is 800, and the stake cap
fraction is $0.90$. ATLAS-SRS uses the same inference and predictor
with uniform sampling.

\paragraph{Results.} \Cref{tab:nonlinear-audits} presents full numerical precision-attainment results for the ATLAS policies, the uncertainty-only ablation policy, and baselines.

\providecommand{\nlauditcell}[3]{%
  \begin{tabular}[t]{@{}r@{}}
  #1\\[-2pt]{\scriptsize[#2, #3]}
  \end{tabular}}

\begin{table}[h!]
\centering
\caption{Median [25th, 75th percentiles] audits to first attain the
stated precision at 95.5\% nominal confidence over 100 paired paths.
Precision is evaluated every 25 audits. $>B$ denotes nonattainment
within budget: $B=3{,}000$ for Kenya ATLAS and uncertainty-only,
and $B=10{,}000$ otherwise. Quantiles include all paths, with
unreached targets treated as censored. Counts are rounded to
the nearest audit.}
\label{tab:nonlinear-audits}
\makebox[\linewidth][c]{%
\footnotesize
\setlength{\tabcolsep}{2.5pt}
\renewcommand{\arraystretch}{1.15}
\begin{tabular}{@{}lllrrrrr@{}}
\toprule
& & & \multicolumn{2}{c}{5\% relative half-width}
      & \shortstack{25\% relative\\half-width}
      & \multicolumn{2}{c}{$\pm 2$ percentage points} \\
\cmidrule(lr){4-5}\cmidrule(lr){6-6}\cmidrule(l){7-8}
\textbf{Method}
& \textbf{Sampling strategy}
& \textbf{Guarantee}
& \shortstack{\textbf{Kenya}\\{\scriptsize Hill--Shannon}}
& \shortstack{\textbf{CIFAR-10N}\\{\scriptsize Worst-three precision}}
& \shortstack{\textbf{CIFAR-10N}\\{\scriptsize Worst-three}\\{\scriptsize taxonomic risk}}
& \shortstack{\textbf{ReImageNet}\\{\scriptsize F1}}
& \shortstack{\textbf{CleanCoNLL}\\{\scriptsize F1}} \\
\midrule
\textbf{ATLAS-PAS}
& \begin{tabular}[t]{@{}l@{}}Score-aware\\variance allocation\end{tabular}
& Anytime-valid
& \nlauditcell{\textbf{1,625}}{1,550}{1,700}
& \nlauditcell{\textbf{3,850}}{3,288}{4,431}
& \nlauditcell{\textbf{3,450}}{2,744}{4,594}
& \nlauditcell{\textbf{2,250}}{1,850}{2,625}
& \nlauditcell{\textbf{2,638}}{1,994}{3,594}
\\[3pt]
ATLAS-SRS & Uniform & Anytime-valid
& \nlauditcell{2,550}{2,244}{2,900}
& \nlauditcell{6,550}{5,056}{7,875}
& \nlauditcell{5,588}{4,044}{6,994}
& \nlauditcell{7,638}{5,969}{9,425}
& \nlauditcell{$>$10,000}{$>10$k}{$>10$k}
\\[3pt]
Uncertainty-only
& \begin{tabular}[t]{@{}l@{}}Score-agnostic\\uncertainty allocation\end{tabular}
& Anytime-valid
& \nlauditcell{1,913}{1,825}{2,000}
& \nlauditcell{6,188}{4,969}{7,456}
& \nlauditcell{5,475}{4,369}{7,094}
& \nlauditcell{3,363}{2,706}{4,131}
& \nlauditcell{2,800}{1,969}{3,919}
\\[3pt]
PPR + projection & Uniform & Anytime-valid
& \nlauditcell{$>$10,000}{$>10$k}{$>10$k}
& \nlauditcell{$>$10,000}{$>10$k}{$>10$k}
& \nlauditcell{$>$10,000}{$>10$k}{$>10$k}
& \nlauditcell{7,138}{6,700}{7,350}
& \nlauditcell{$>$10,000}{$>10$k}{$>10$k}
\\[3pt]
MR stitching & Uniform & Anytime-valid
& \nlauditcell{$>$10,000}{$>10$k}{$>10$k}
& \nlauditcell{$>$10,000}{$>10$k}{$>10$k}
& \nlauditcell{$>$10,000}{$>10$k}{$>10$k}
& \nlauditcell{$>$10,000}{$>10$k}{$>10$k}
& \nlauditcell{$>$10,000}{$>10$k}{$>10$k}
\\[3pt]
\bottomrule
\end{tabular}}
\end{table}

\subsubsection{Auditing policies: inference and allocation}

We compare five methods at 95.5\% confidence for each score. All methods
sample without replacement. The three ATLAS policies use the same
predictor and confidence sequence within each experiment. They differ
only in how they choose the next record.

\paragraph{ATLAS confidence sequences.}
After each audit, we use the verified label and its sampling probability
to update the evidence against candidate score values. The construction
depends on the score.

\begin{itemize}
\item \textbf{Hill--Shannon biodiversity score.}
We construct a confidence sequence for Shannon entropy as in \cref{sec:entropy}. For the upper
endpoint, the predicted label proportions determine a linear upper
bound on entropy before each audit. We use this bound to test whether
entropy is at least a candidate value $h$. For the lower endpoint, we
test whether entropy is at most $h$. We reject this hypothesis only
when the evidence reaches the rejection threshold for every label
composition consistent with the audits so far. Appendix~G shows how
to check this condition by convex optimization. Exponentiating the two
entropy endpoints gives the interval for Hill--Shannon diversity.

\item \textbf{Worst-three precision and taxonomic risk.}
For each of the 120 triples of CIFAR-10N classes, average precision is
linear in the confusion matrix. Worst-three precision is the minimum
of these averages. To establish that the minimum exceeds $m$, we
require every triple's upward test to cross its threshold. To establish
that the minimum is below $m$, we update one downward test using the
triple predicted \emph{before each audit} to have the lowest precision.
This is valid because, if the minimum is at least $m$, every selected
triple has average precision at least $m$. Worst-three taxonomic risk
is a maximum, so the directions reverse: all triples determine the
upper endpoint, and the triple predicted to have the highest risk is
used to update the lower endpoint.

\item \textbf{F1.}
Write the score as $\sum_i u_i/\sum_i w_i$, where $u_i$ and $w_i$ are
record $i$'s contributions to the numerator and denominator. For each
candidate $f$, the score equals $f$ exactly when the population mean
of $u_i-f w_i$ is zero. We apply the two one-sided e-processes of
\cref{app:scores} to this linear score and retain $f$ while neither process
has reached $2/\alpha$.
\end{itemize}

\paragraph{Uncertainty-only allocation.}
This policy favors strata whose verified labels are difficult to
predict, regardless of how those labels affect the score. Let
$q_{gk,t}$ be the predicted probability of verified label $k$ in
stratum $g$ and define $s_{g,t}^2=1-\sum_k q_{gk,t}^2$. This quantity
is zero when the predictor assigns all its mass to one label and
increases as the predictions become more dispersed. Using the same
uniform sampling weight $\epsilon$ as PAS, we select stratum $g$ with
probability
\[
\pi_{g,t}
=
\epsilon\frac{M_{g,t}}{\sum_h M_{h,t}}
+
(1-\epsilon)
\frac{M_{g,t}s_{g,t}}{\sum_h M_{h,t}s_{h,t}}.
\]
If the second denominator is zero, we sample uniformly. Relative to
this policy, PAS reduced median audit counts by 15\% for diversity,
38\% for worst-three precision, 37\% for worst-three taxonomic risk,
33\% for ReImageNet F1, and 6\% for CleanCoNLL F1. For the binary
mislabel indicator, $s_{g,t}$ is proportional to the predicted
standard deviation of that indicator, so the two allocation rules
coincide.

\paragraph{PPR with projection.}
We first construct simultaneous confidence sequences for the counts
that determine each score. For each count, we use a uniform working
prior over possible population counts and update it using the
hypergeometric likelihood. At each audit, we retain candidate counts
whose prior-to-posterior probability ratio has stayed below the
reciprocal of the error budget assigned to that count. The budgets
sum to $\alpha$, so all retained count intervals cover simultaneously
with probability at least $1-\alpha$. We then bound the score over
the counts that remain possible.

For diversity, the verified-label counts must also sum to $N$. We
bound entropy over the resulting set of label proportions using chords
and tangents, then exponentiate the endpoints. For worst-three
precision, we bound each class's precision and average the three
smallest lower bounds and the three smallest upper bounds. For
taxonomic risk, we use the row-label count bounds to bound each
class's risk and average the three largest bounds on each side.
For F1, let $A$ be the known number of records submitted as a target
class. Among records whose submitted and verified labels differ, let
$\mathrm{FP}$ count those submitted as a target class and $\mathrm{FN}$
count those verified as a target class. A correction between two target
classes contributes to both counts. The number of true positives is
therefore $A-\mathrm{FP}$, giving
\[
F_1=\frac{2(A-\mathrm{FP})}{2A-\mathrm{FP}+\mathrm{FN}}.
\]
We construct a confidence sequence for each count at error level
$\alpha/2$. Since $F_1$ decreases as either count increases, substituting
both upper bounds gives the lower F1 endpoint, and substituting both
lower bounds gives the upper endpoint.

\paragraph{MR stitching.}
We follow the stitching schedule and constants of Manole and Ramdas
(2023). \footnote{See \url{https://github.com/sshekhar17/nonparametric-testing-by-betting}
for an implementation of the Manole--Ramdas method.}
Within each time window, we use concentration bounds for uniform sampling
without replacement. For diversity, we use the convexity of negative
entropy to bound Shannon entropy, then exponentiate the endpoints.
For CIFAR-10N, we bound the average score for each triple of classes
and take the minimum for worst-three precision or the maximum for
worst-three taxonomic risk. For F1, we retain each candidate $f$ for which the interval for the population mean of $u_i-f w_i$ contains zero.
At each update, we intersect the interval with all earlier intervals
and with the range of scores consistent with the labels verified so far.

\paragraph{Empirical coverage.}

Table~\ref{tab:nonlinear-coverage} counts a path as covered if its interval contains the true
score at every evaluated update through the full audit budget,
including updates after the target precision is first reached.
All paths contribute to the efficiency summaries, including those
with a coverage failure.

\begin{table}[h]
\centering
\caption{Whole-path coverage counts out of 100 paths.
Budgets match \cref{tab:nonlinear-audits}.}
\label{tab:nonlinear-coverage}
\footnotesize
\setlength{\tabcolsep}{5pt}
\begin{tabular}{@{}lrrrrr@{}}
\toprule
\textbf{Score} & \textbf{PAS} & \textbf{SRS}
& \textbf{Uncertainty-only} & \textbf{PPR} & \textbf{MR} \\
\midrule
Kenya diversity & 100 & 100 & 100 & 99 & 100 \\
CIFAR worst-three precision & 99 & 100 & 100 & 99 & 100 \\
CIFAR worst-three taxonomic risk & 100 & 100 & 100 & 100 & 100 \\
ReImageNet F1 & 96 & 97 & 95 & 100 & 100 \\
CleanCoNLL F1 & 95 & 89 & 95 & 100 & 100 \\
\bottomrule
\end{tabular}
\end{table}

\subsection{Synthetic Data: A Priority Mechanism Study}
\label{app:synthetic-experiment-settings}

The synthetic experiment isolates the acquisition advantage produced by
heterogeneous score relevance. Its confirmation population contains 20
equal-sized submitted-label rows with 250 records per row, for \(N=5{,}000\).
Four rows have equal numbers of records with score contributions 0 and 1.
The other 16 rows each have two possible verified destinations, both with
score contribution \(0.25\). Their verified labels remain uncertain, but that
uncertainty does not affect the chosen score.

Let $D^{\mathrm{syn}}_{rs}$ denote the fixed score contribution associated with submitted row $r$ and verified destination $s$. The target is the population-wide linear risk
\[
g(\mathbf C)=\frac{1}{N}\sum_{r,s}D^{\mathrm{syn}}_{rs}C_{rs}.
\]
The score equals $0.30$ in the confirmation population.
To isolate the allocation mechanism, PAS uses the known population
row standard deviations in the remaining-size-aware Neyman rule
\[
\pi^{\mathrm{PAS}}_{r,t}\propto M_{r,t}\sigma_r,
\]
where $\sigma_r=0.5$ for the four score-variable rows and zero
for the constant-contribution rows. These standard deviations
remain fixed throughout the experiment, and no SRS floor is used.
Rows with constant score contributions receive no PAS audits.
SRS remains uniform over all unaudited records. Both policies use
the same prediction-assisted linear betting confidence sequence.

\begin{table}[h!]
\centering
\caption{Frozen synthetic-experiment settings.}
\label{tab:synthetic-settings}
\small
\setlength{\tabcolsep}{5pt}
\begin{tabular}{@{}lr@{}}
\toprule
\textbf{Setting} & \textbf{Value} \\
\midrule
Confirmation rows & 20 \\
Score-variable rows & 4 \\
Records per row & 250 \\
Population size & \(5{,}000\) \\
Confirmation paths & 100 \\
Confirmation seeds & \(2026101000,\ldots,2026101099\) \\
Maximum audits & 800 \\
Checkpoint spacing & 25 \\
Confidence level & \(95.5\%\) \\
Prior strength & 2 \\
Stake effective horizon & 800 \\
Stake cap fraction & \(0.90\) \\
Candidate-grid spacing & \(5\times10^{-4}\) \\
SRS component in PAS & 0 \\
Allocation variances & Frozen population row variances \\
\bottomrule
\end{tabular}
\end{table}

The scaling screen fixes four score-variable rows and varies the total number
of rows over
\[
R\in\{8,12,20,32,48\}.
\]
Each configuration uses five development seeds and an 800-audit budget. This screen tests the predicted increase in PAS advantage as the number of score-invariant nuisance rows grows. Confirmation fixes \(R=20\) and uses the 100 fresh paths reported above.

\begin{figure}[h]
    \centering
    \includegraphics[width=\linewidth]
    {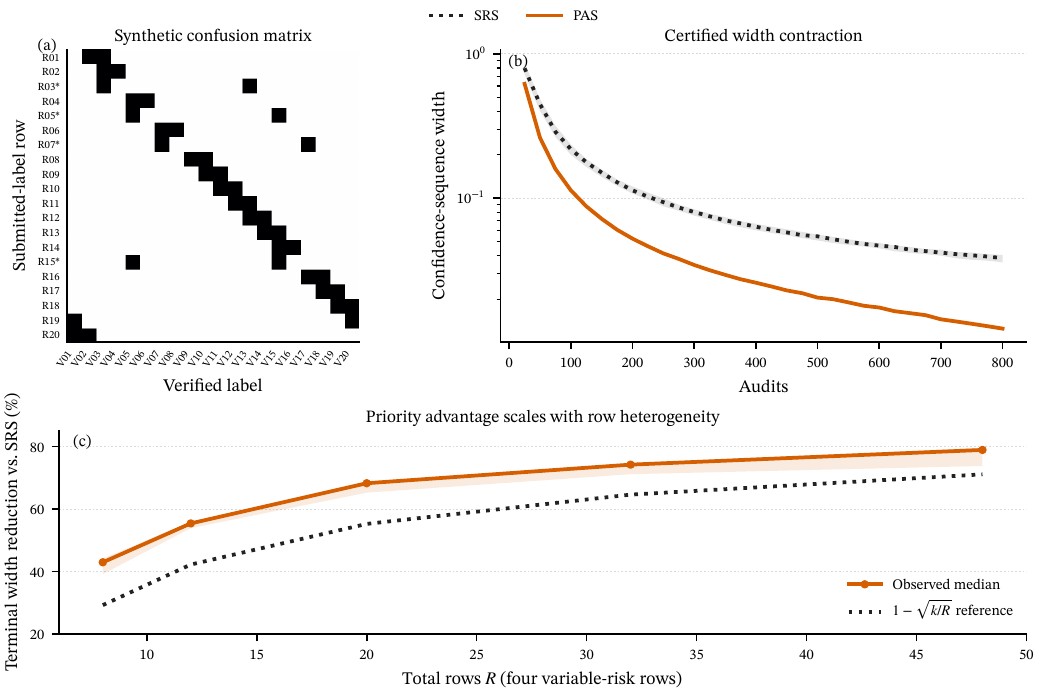}
    \caption{
Controlled synthetic experiment isolating the benefit of priority-aware sampling.
\textbf{(a)}\ Confusion matrix; asterisks identify the four rows whose possible verified labels have different score contributions.
\textbf{(b)}\ Median confidence-sequence width under PAS and SRS, with interquartile ranges.
\textbf{(c)}\ Terminal PAS width reduction as the total number of rows \(R\) increases while the number of score-variable rows remains fixed at \(k=4\). The dotted curve shows the idealized \(1-\sqrt{k/R}\) reference: SRS allocates only a fraction \(k/R\) of its audits to the score-variable rows, whereas PAS concentrates on them, and confidence-sequence width scales approximately with the inverse square root of the informative audit count.
}
    \label{fig:synthetic-priority-upside}
\end{figure}

\subsection{Additional Mislabel-Rate Comparators}
\label{app:linear-comparators}

The default prior--posterior-ratio confidence sequence is evaluated on the same
100 SRS audit orders used by our betting confidence sequence. It uses a uniform
\(\mathrm{BetaBinomial}(N,1,1)\) working prior on the population error count
and reports the running intersection of its candidate sets.

The Wald reference uses the true population error variance, no
finite-population correction, and nominal confidence \(95.5\%\). It is an
oracle-variance, fixed-sample asymptotic reference. The Serfling reference is
the corresponding finite-population, fixed-sample bound. Wald and Serfling
therefore provide benchmark sample sizes rather than anytime-valid stopping
rules.

The Chebyshev reference uses the Bernoulli variance bound \(1/4\)
with the finite-population correction. Its required audit count is
computed analytically. Like the Serfling reference, it is a
fixed-sample benchmark rather than an anytime-valid stopping rule.

Fixed-sample benchmark counts are computed from their deterministic
half-width formulas and therefore have no across-path interquartile
range. In particular, the CleanCoNLL Serfling result uses the
unclipped radius, which reaches $0.02$ at 4,670 audits.

We evaluate two implementations of \citet{shekhar2023risk}, aligned
with the authors' code at commit \texttt{a834e459}.
For mislabel-rate estimation, all records have equal weight and the
audited outcome is the binary error indicator. Prop-MS samples in
proportion to predicted error probability. Prop-M + CV samples
uniformly and uses the error predictor as a centered control variate.
Both variants receive the same initially flat predictor, updated
only from paid audits.

We retain the reference approximate-Kelly betting rules and use
\(\alpha=0.045\) with candidate-grid spacing \(0.0005\).
The implementation supplies predictable, audit-updated side
information and conservatively certifies rejection over grid cells
for continuous-parameter reporting. Both variants use the same
100 seeds and 5,000-audit budget as ATLAS.

Whole-path coverage for Prop-M + CV is 90/100 on Kenya,
98/100 on CIFAR-10N, and 97/100 on both ReImageNet and CleanCoNLL,
evaluated at every audit through the full 5,000-audit budget.
Coverage failures are retained in all reported summaries.
On CleanCoNLL, its smaller median stopping count is accompanied
by target attainment on 75/100 paths, compared with 100/100
for ATLAS-PAS.

\subsection{Sensitivity Studies and Robustness}
\label{app:sensitivity}

\subsubsection{Sensitivity to the SRS mixing weight $\varepsilon$}
\label{app:eps-sensitivity}

ATLAS-PAS mixes its priority-aware allocation with simple random sampling:
\[
\pi_{r,t}^{(\varepsilon)}
=
(1-\varepsilon)\pi_{r,t}^{\mathrm{Neyman}}
+
\varepsilon\pi_{r,t}^{\mathrm{SRS}},
\]
where the Neyman component is proportional to remaining stratum
size times predicted error standard deviation.

The mixing weight $\epsilon$ ensures continued exploration of every stratum. We evaluated
\[
\epsilon\in\{0.05,0.1,0.2,\ldots,0.9,1.0\},
\]
where $\epsilon=1$ is ATLAS-SRS and $\epsilon=0.3$ is the default used in the mislabel-rate experiments. All other settings were fixed. Each value was evaluated on the same 100 seeds and within-stratum record orderings for Kenya, CIFAR-10N, ReImageNet, and CleanCoNLL.

For each path, we recorded the first audit at or after 100
at which the nonempty 95.5\% confidence interval had half-width
at most 2 percentage points. Each path continued through 5,000 audits. Figure~\ref{fig:epsilon-sensitivity} reports the median and interquartile range of these first-crossing counts. To place datasets with different audit requirements on one scale, each dataset is divided by its largest median over the $\epsilon$ grid.

\paragraph{Results.}The default $\epsilon=0.3$ is optimal on CIFAR-10N and lies within 22 audits of the best tested value on Kenya and within 6 audits on ReImageNet. At the default, the median audit counts are 603, 790, and 1,394, compared with 814, 980, and 1,539 under SRS. CleanCoNLL favors greater exploration, with its best result near $\epsilon=0.7$, although the default still improves slightly over SRS. This behavior is consistent with its highly imbalanced label distribution: early priority estimates can concentrate audits too strongly, while additional random exploration stabilizes the allocation. Overall, the results show that the main gains persist across a broad range of mixing weights and do not strongly depend on dataset-specific tuning of $\epsilon$.

\begin{figure*}[t]
    \centering
    \includegraphics[width=0.88\textwidth]
    {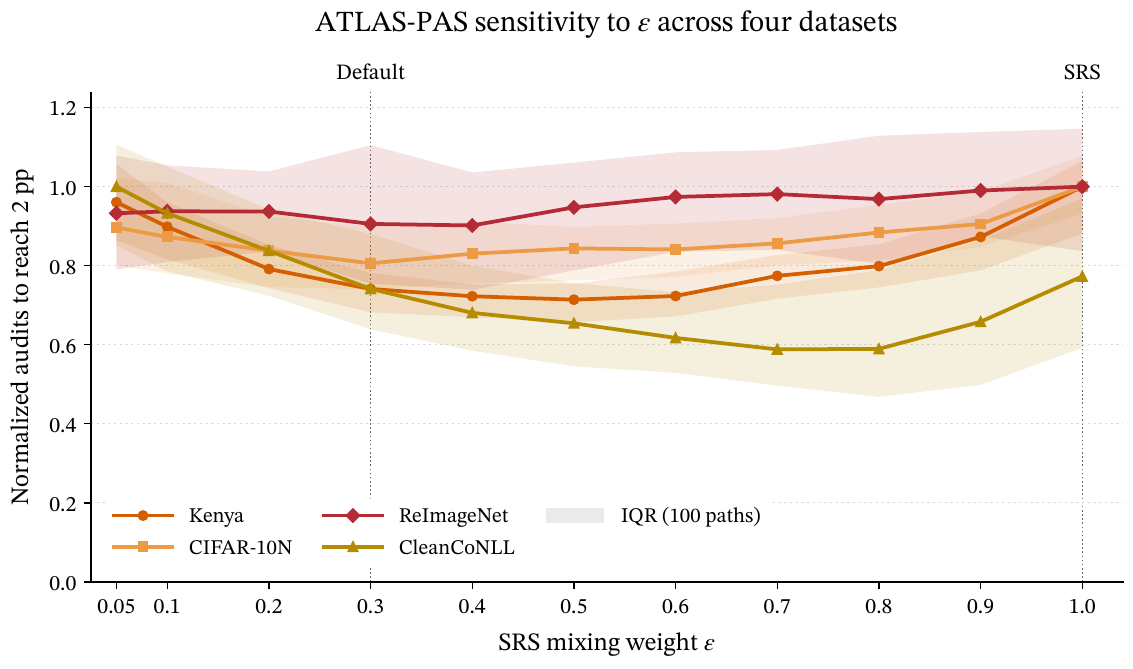}
    \caption{
    Sensitivity of mislabel-rate certification to the SRS mixing weight
    $\epsilon$. Points show the median number of audits required to reach a
    2-percentage-point half-width, and shaded regions show the 25th--75th
    percentiles over 100 paired paths. Results are normalized separately for
    each dataset by its largest median over the tested grid. The dotted lines
    mark the default $\epsilon=0.3$ and ATLAS-SRS at $\epsilon=1$.
    }
    \label{fig:epsilon-sensitivity}
\end{figure*}

\subsubsection{Sensitivity to predictive initialization}
\label{app:prior-sensitivity}

PAS uses a working predictive distribution to estimate which strata are most uncertain. Its formal guarantee holds for every predictable initialization, and this experiment measures how those initializations affect efficiency and verification of empirical coverage. We varied the initialization in the Kenya mislabel-rate and Hill--Shannon diversity experiments while keeping the confidence sequences, sampling rules, and all other parameters fixed.

For mislabel rate, the initial probability that a submitted label is correct was varied from 0\% to 100\% in increments of 10 percentage points. The default is 50\%. For diversity, we used
\[
(1-w)e_{y_r}+w\bar p_{t-1},
\]
where $e_{y_r}$ places all mass on the submitted label and
$\bar p_{t-1}$ is the pooled distribution formed from previously
audited labels, initialized uniformly. Thus $w$ controls the
initial uniform mixture and subsequently the weight on the
audit-updated pooled distribution.

We ran 100 paired paths at each of the 11 initialization settings
for each score, with budgets of 5,000 audits for mislabel rate
and 3,000 for diversity. This gives 4,400 policy paths in total,
reusing the same 100 seeds across initializations within each
score. Precision is evaluated after every audit for mislabel
rate, with attainment eligible from audit 100, and every
25 audits for diversity. Targets are a 2-percentage-point
half-width and a 5\% relative half-width, respectively.
Whole-path coverage is the fraction of paths whose intervals
contain the true score at every checked update through the
full budget, including updates after target attainment.

\paragraph{Results.}For mislabel rate, PAS requires between 603 and 631 median
audits across initial correctness probabilities from 20\%
to 80\%. Extreme initializations are less efficient: the
0\% and 100\% settings require 1,566 and 2,352 median audits,
respectively, and perform worse than SRS under the same
initializations. For diversity, PAS requires between 1,550
and 1,625 median audits for mixture weights from 10\% through
100\%, while the all-submitted-label initialization requires
1,888. Whole-path coverage ranges from 94--97\% for linear PAS
and 99--100\% for diversity PAS. These results support stable
efficiency over a broad range of initializations while also
showing the cost of severely misspecified initial predictions.

\begin{figure*}[t]
    \centering
    \includegraphics[width=0.94\textwidth]
    {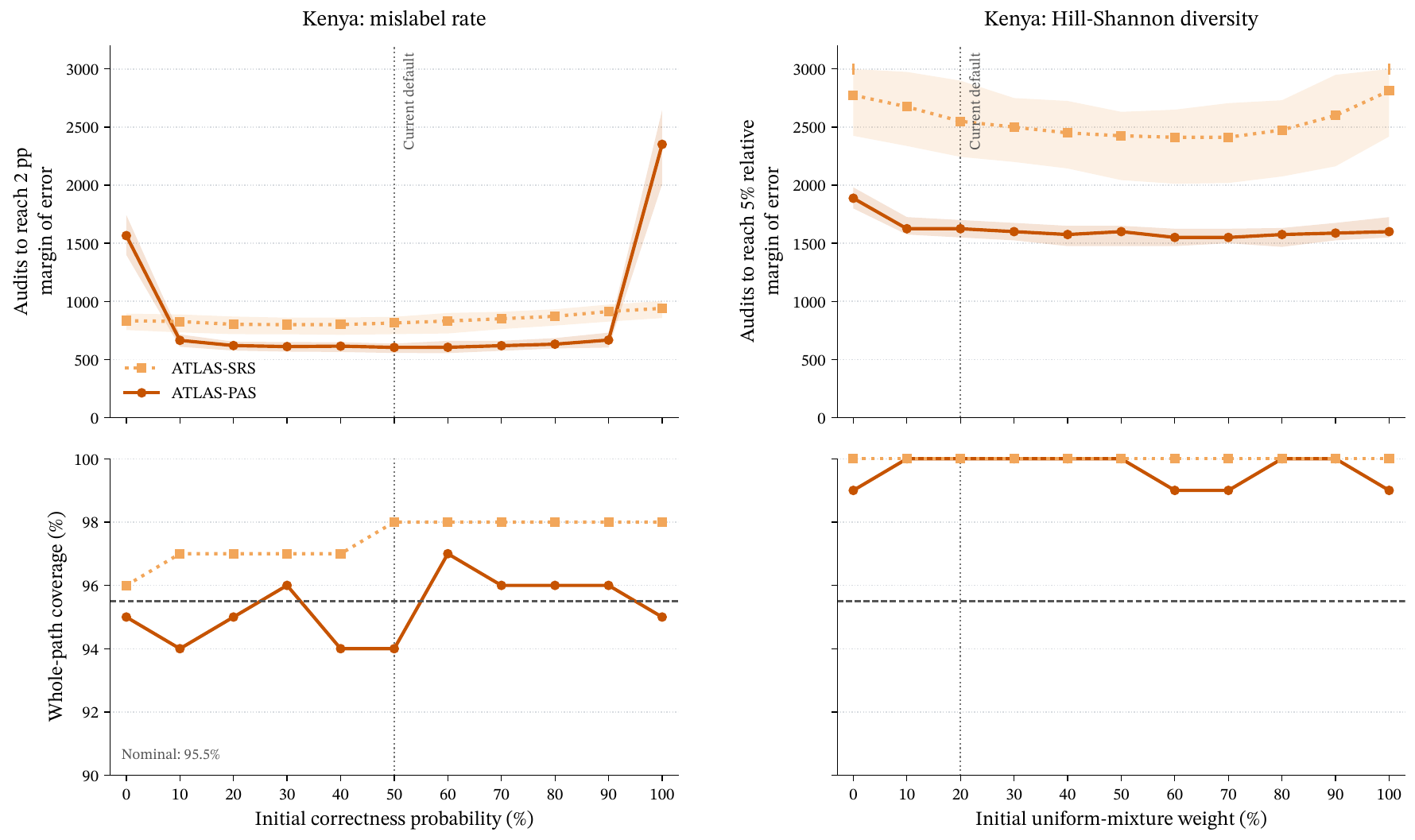}
    \caption{
    Sensitivity to predictive initialization on the Kenya dataset. The top row
    shows median audits to the target precision, with 25th--75th percentile
    bands over 100 paired paths. The bottom row shows empirical whole-path
    coverage through the full audit budget. The left column varies the initial
    correctness probability for mislabel-rate certification; the right column
    varies the uniform-mixture weight for Hill--Shannon diversity. Vertical
    dotted lines mark the defaults, and the horizontal dashed line marks
    95.5\% nominal coverage.
    }
    \label{fig:initialization-sensitivity}
\end{figure*}

\subsubsection{Sensitivity to predictor misspecification}
\label{app:predictor-misspecification}

We evaluate how errors in the predictive distribution affect the efficiency of
ATLAS-PAS on the Kenya mislabel-rate task. Each condition uses 100 paired audit
paths, the same 30\% SRS exploration component, and the same 95.5\% confidence
sequence as the main experiment. The endpoint is the first audit after audit
100 at which the interval half-width reaches two percentage points. PAS and SRS
use the same predictor transformation within each condition, applied to their
respective audit histories. All transformations are predictable and therefore
preserve the validity of the confidence sequence.

Let \(p_{r,t}\) denote the error probability predicted for stratum \(r\) before
audit \(t\). Calibration is varied using
\[
q_{r,t}
=
\frac{p_{r,t}^{\,k}}
     {p_{r,t}^{\,k}+(1-p_{r,t})^{\,k}},
\]
where \(k<1\) produces underconfident predictions and \(k>1\) produces
overconfident predictions while preserving the uncertainty ranking across
strata. We also test a constant predictor, \(q_{r,t}=0.5\), and a deliberately
adversarial transformation that reverses the learned ranking of
\(p_{r,t}(1-p_{r,t})\). The reversal is applied for the first 25, 50, or 100
audits, or throughout the complete audit. After a temporary reversal, the
ordinary learned predictor resumes while retaining the preceding audit history.


\begin{table}[t]
\centering
\caption{Sensitivity of ATLAS to predictor misspecification for mislabel-rate certification on the Kenya biodiviersity dataset. Audit counts are medians and [25th,75th] percentiles over 100 paths to reach a two-percentage-point margin of error. Negative savings indicate that PAS was slower than SRS. Coverage entries give PAS/SRS counts, each out of 100, for containing the true value
at every audit through 5,000. All paths reach the target except 35 PAS paths under ranking reversal.}
\label{tab:predictor-misspecification}

\footnotesize
\setlength{\tabcolsep}{3pt}
\renewcommand{\arraystretch}{1.12}

\newcommand{\stressauditcell}[3]{%
  \begin{tabular}[t]{@{}r@{}}
    #1\\[-1pt]
    {\scriptsize[#2, #3]}
  \end{tabular}%
}

\begin{tabularx}{\linewidth}{@{}
  >{\raggedright\arraybackslash}X
  >{\raggedleft\arraybackslash}p{0.18\linewidth}
  >{\raggedleft\arraybackslash}p{0.18\linewidth}
  >{\centering\arraybackslash}p{0.10\linewidth}
  >{\centering\arraybackslash}p{0.16\linewidth}
@{}}
\toprule
\textbf{Predictor}
& \shortstack[b]{\textbf{PAS}\\\textbf{audits}}
& \shortstack[b]{\textbf{SRS}\\\textbf{audits}}
& \shortstack[b]{\textbf{PAS}\\\textbf{savings}}
& \shortstack[b]{\textbf{Whole-path}\\\textbf{coverage}} \\
\midrule

Default
& \stressauditcell{603}{555}{636}
& \stressauditcell{813.5}{717.5}{866.2}
& \(25.9\%\)
& 94/98 \\
\addlinespace[3pt]

Underconfident (\(k=0.5\))
& \stressauditcell{628.5}{586.8}{666}
& \stressauditcell{852}{743.2}{909.2}
& \(26.2\%\)
& 93/94 \\
\addlinespace[3pt]

Overconfident (\(k=2\))
& \stressauditcell{787}{715.2}{891}
& \stressauditcell{854}{753}{914}
& \(7.8\%\)
& 96/96 \\
\addlinespace[3pt]

Overconfident (\(k=4\))
& \stressauditcell{1,260}{1,076.5}{1,467}
& \stressauditcell{883}{805.8}{972.8}
& \(-42.7\%\)
& 96/97 \\
\addlinespace[3pt]

Constant (\(q_{r,t}=0.5\))
& \stressauditcell{2,449.5}{1,981.2}{2,879.5}
& \stressauditcell{2,449.5}{1,981.2}{2,879.5}
& \(0.0\%\)
& 91/91 \\
\addlinespace[5pt]

Reversed: first 25 audits
& \stressauditcell{854.5}{718.8}{911.5}
& \stressauditcell{846}{774.2}{908.5}
& \(-1.0\%\)
& 94/97 \\
\addlinespace[3pt]

Reversed: first 50 audits
& \stressauditcell{950}{851}{1,038.5}
& \stressauditcell{870.5}{807}{926.2}
& \(-9.1\%\)
& 96/95 \\
\addlinespace[3pt]

Reversed: first 100 audits
& \stressauditcell{1,085}{1,016.2}{1,169}
& \stressauditcell{924}{870.8}{986.8}
& \(-17.4\%\)
& 95/93 \\
\addlinespace[3pt]

Reversed throughout
& \stressauditcell{4,441}{3,690.8}{\(>5{,}000\)}
& \stressauditcell{1,775}{1,578.5}{2,012}
& \(-150.2\%\)
& 99/98 \\

\bottomrule
\end{tabularx}
\end{table}

\paragraph{Results.}Moderate misspecification degrades efficiency without eliminating the benefit:
doubling the predicted log-odds reduces the PAS savings from 25.9\% to 7.8\%.
Underconfidence has little effect, whereas stronger overconfidence makes PAS
42.7\% slower than SRS despite preserving the uncertainty ranking. With a
constant predictor, PAS reduces exactly to allocation by remaining stratum
mass and therefore matches SRS.

Incorrect early rankings have a persistent cost. Reversing the learned
uncertainty ranking for the first 25, 50, and 100 audits changes the relative
savings to \(-1.0\%\), \(-9.1\%\), and \(-17.4\%\), respectively. Persistent
reversal is substantially more damaging, with only 65 of 100 PAS paths
attaining the target within budget. These results characterize the practical
limit of the efficiency claim: validity is retained under predictable
misspecification, but audit savings depend on the predictive distribution
providing useful estimates of score-relevant uncertainty, particularly during
the early audit. The whole-path coverage counts are included as finite-run
diagnostics; the formal coverage guarantee follows from predictability rather
than from these empirical frequencies.

\section{Computational Environment}
\label{app:computational-environment}

All experiments were executed on a local machine with 15 CPU cores and 24\,GB of unified memory. Computation was
CPU-only; no GPU or external computing cluster was used. Independent
score-policy-seed tasks were distributed across four local worker processes.

The final runs used macOS 26.5 and Python 3.14.7. Principal numerical and
reporting dependencies were NumPy 2.5.2, SciPy 1.18.1, pandas 3.0.5, and
Matplotlib 3.11.1. Each worker wrote an independent result shard, and
aggregation began only after all expected shards for that experiment had
completed.

%% file: appendices/confidence_sequences_appendices.tex
\section{Linear scores}
\label{app:linear}

\begin{lemma}[Ville's inequality \cite{ville1939etude}]
\label{lem:ville}
Let $(W_t)_{t\geq0}$ be a nonnegative supermartingale with respect to $(\mathcal F_t)$ with
$W_0=1$. Then $\Pr\{\exists t:W_t\geq c\}\leq1/c$ for every $c>0$.
\end{lemma}

\paragraph{Predictions.}
Any predictable $\mathbf{q}_{g,t}$ keeps Lemma~\ref{lem:unbiased} valid; good predictions only
reduce the variance \eqref{eq:varZ}. We use a two-level shrinkage estimator, built
only from the audits before audit $t$. Let $\mathbf{A}_{r,t}(s)$ and $A_{g,t}(s)$ be the numbers
of audited records in row $r$ and in stratum $g$ that were verified as $s$, and
$n_{r,t}=\sum_s\mathbf{A}_{r,t}(s)$, $n_{g,t}=\sum_s\mathbf{A}_{g,t}(s)$ the numbers audited. Then
\begin{align}
  \bar g_t(s)&=\frac{1/L+\sum_r\mathbf{A}_{r,t}(s)}{1+(t-1)}
    &&\text{(pooled verified labels)}\nonumber\\
  \mathbf{q}_{r,t}&=\frac{\mathbf{A}_{r,t}+\kappa_r\{\beta \mathbf e_r+(1-\beta)\bar g_t\}}{n_{r,t}+\kappa_r}
    &&\text{(submitted label $r$)}\label{eq:pred-row}\\
  \mathbf{q}_{g,t}&=\frac{\mathbf{A}_{g,t}+\kappa_c\,\mathbf{q}_{r(g),t}}{n_{g,t}+\kappa_c}
    &&\text{(stratum $g$ within row $r(g)$)},\label{eq:pred-stratum}
\end{align}
where $\mathbf e_r$ is the unit vector at $r$. Before any audit in row $r$, the prediction puts
weight $\beta$ on ``the submitted label is correct'' and spreads the rest over the labels
verified so far. As audits accumulate, each row moves toward its own verified-label
frequencies, and each stratum toward its own frequencies within its row. Each level is
the posterior mean of a Dirichlet--multinomial model whose prior is centered on the level
above. This is the same structure as the Bayesian model of \cref{sec:bayes}, but here it
is used only for prediction, never for inference. On the Kenya data we use $\beta=0.8$,
$\kappa_r=2$ and $\kappa_c=1$, fixed before the audit and not tuned.

\subsection{Proofs for \cref{sec:linear}}
\label{app:linear-proofs}

\begin{proof}[Proof of \cref{lem:unbiased}]
\emph{Unbiasedness.} Given $\mathcal F_{t-1}$, $G_t=g$, and \(J_t\in[N]\) as the record selected for expert audit at time \(t\), $J_t$ is uniform on $g\setminus\mathcal A_{t-1}$, so
$\E[D_{y_{J_t}V_t}\mid\mathcal F_{t-1},G_t=g]=\mu_{g,t}$. Averaging over $G_t$,
\begin{align*}
  \E\Bigl[\frac{M_{G_t,t}}{N\pi_{G_t,t}}\bigl(D_{y_{J_t}V_t}-\mathbf{q}_{G_t,t}^\top D_{r(G_t)\cdot}\bigr)\Bigm|\mathcal F_{t-1}\Bigr]
  &=\sum_{g:M_{g,t}>0}\pi_{g,t}\frac{M_{g,t}}{N\pi_{g,t}}\bigl(\mu_{g,t}-\mathbf{q}_{g,t}^\top \mathbf d_{r(g)}\bigr)\\
  &=\frac1N\sum_gM_{g,t}\bigl(\mu_{g,t}-\mathbf{q}_{g,t}^\top \mathbf d_{r(g)}\bigr).
\end{align*}
The condition $\pi_{g,t}>0$ whenever $M_{g,t}>0$ lets the probabilities cancel.
Adding $\hat\theta_{t-1}$, the predicted terms cancel and \eqref{eq:split} gives
$\theta^\star$. \emph{Variance.} Only the correction is random. Its conditional
second moment is
$\sum_g\pi_{g,t}(M_{g,t}/N\pi_{g,t})^2\E[(D_{y_{J_t}V_t}-\mathbf q^\top_{g,t}\mathbf d_{r(g)})^2\mid G_t=g]
=N^{-2}\sum_g(M^2_{g,t}/\pi_{g,t})(\sigma^2_{g,t}+b^2_{g,t})$. Subtracting the square
of its mean, $N^{-1}\sum_gM_{g,t}b_{g,t}$, gives \eqref{eq:varZ}.
\end{proof}

\begin{proof}[Proof of \cref{prop:neyman}]
With exact predictions $b_{g,t}=0$ and \eqref{eq:varZ} is $N^{-2}\sum_g\gamma_g^2/\pi_g$
with $\gamma_g=M_{g,t}\sigma_{g,t}$. This is convex on
$\{\pi\geq\ell,\sum_g\pi_g=1\}$, and the KKT conditions give
$\pi_g=\max\{\ell_g,\gamma_g/\nu\}$. Without floors, Cauchy--Schwarz gives
$\sum_g\gamma_g^2/\pi_g\geq(\sum_g\gamma_g)^2$ with equality when $\pi_g\propto\gamma_g$. For refinement, split $g$ into $g_1,\dots,g_J$ with shares $\rho_j$; by the law of total
variance and Jensen, $\sigma^2_g\geq\sum_j\rho_j\sigma^2_{g_j}\geq(\sum_j\rho_j\sigma_{g_j})^2$,
so $M_g\sigma_g\geq\sum_jM_{g_j}\sigma_{g_j}$.
\end{proof}

\paragraph{The range of one observation.}
The stake caps need predictable bounds valid for every stratum that could be drawn and every unobserved label:
\[
  \underline{Z}_t=\hat\theta_{t-1}+\min_{g:M_{g,t}>0}\frac{M_{g,t}}{N\pi_{g,t}}\Bigl(\min_vD_{r(g)v}-\mathbf q^\top_{g,t}D_{r(g)\cdot}\Bigr),
\]
\[
  \overline{Z}_t=\hat\theta_{t-1}+\max_{g:M_{g,t}>0}\frac{M_{g,t}}{N\pi_{g,t}}\Bigl(\max_vD_{r(g)v}-\mathbf q^\top_{g,t}D_{r(g)\cdot}\Bigr).
\]
For the error rate, $\min_vD_{rv}=0$ and $\max_vD_{rv}=1$ in every row.

\begin{proof}[Proof of \cref{thm:linear-cs}]
Consider $E^+_t$ at $m=\theta^\star$. If $\theta^\star>\underline{Z}_t$ the stake is at most
$0.9/(\theta^\star-\underline{Z}_t)$, so $1+\lambda^+_t(Z_t-\theta^\star)\geq1-0.9=0.1$; if
$\theta^\star\leq\underline{Z}_t$ the factor is at least one. By \cref{lem:unbiased} its
conditional mean is one. So $E^+_t(\theta^\star)$ is a nonnegative martingale with
$E^+_0=1$, and \cref{lem:ville} gives $\Pr\{\exists t:E^+_t(\theta^\star)\geq2/\alpha\}\leq\alpha/2$.
The same holds for $E^-_t(\theta^\star)$ with the cap $0.9/(\overline{Z}_t-\theta^\star)$. The
sequence misses $\theta^\star$ only if one of the two reaches $2/\alpha$.
\end{proof}

\paragraph{Evaluating the sequence on a grid.}
In the experiment, candidates are evaluated on the grid $m_j=j/2000$, $j=0,\dots,2000$. On each
interval $[m_j,m_{j+1}]$ each side's stake is held constant at the more constrained
endpoint: $\lambda^+$ capped at $0.9/(m_{j+1}-\underline{Z}_t)$ and 
$\lambda^-$ at
$0.9/(\overline{Z}_t-m_j)$. Then $1+\lambda^+(Z_t-m)$ decreases and $1-\lambda^-(Z_t-m)$
increases in $m$ on the interval, so if $\theta^\star\in[m_j,m_{j+1}]$ is not
rejected, neither are $m_{j+1}$ by the upward test nor $m_j$ by the downward test.
Retaining every grid interval whose right end survives the upward test and whose
left end survives the downward test therefore preserves coverage exactly, and
differs from the hull of surviving grid points by at most one grid step (only 0.05
percentage points).

\paragraph{Choosing the stakes.}
Within the caps, the stake decides how fast evidence accumulates. Suppose the
observations have mean $m+\delta$ and variance $v$. Expanding
$\log(1+x)\approx x-x^2/2$, each review adds about $\lambda\delta-\lambda^2v/2$ to
$\log E^+_t(m)$, so after $n$ reviews the candidate $m$ is ruled out once
$n(\lambda\delta-\lambda^2v/2)\geq\log(2/\alpha)$, that is, once
$\delta\geq\log(2/\alpha)/(n\lambda)+\lambda v/2$. The stake that makes this
smallest is
\[
  \lambda=\sqrt{\frac{2\log(2/\alpha)}{n\,v}},
  \qquad\text{giving}\qquad
  \delta=\sqrt{\frac{2v\log(2/\alpha)}{n}} .
\]
A large stake is safe when the observations are quiet and dangerous when they are
noisy. Neither $v$ nor the number of reviews is known in advance, so we use the
predictable plug-in rule of \citet{waudbysmith2024estimating} in its fixed-horizon form:
we replace $v$ by a running estimate computed from past reviews only, and $n$ by a
planned horizon $n^\star$,
\begin{equation}
  \lambda_t=\sqrt{\frac{2\log(2/\alpha)}{n^\star\hat v_t}},\qquad
  \lambda^+_t(m)=\min\Bigl\{\lambda_t,\frac{0.9}{m-\underline{Z}_t}\Bigr\},\qquad
  \lambda^-_t(m)=\min\Bigl\{\lambda_t,\frac{0.9}{\overline{Z}_t-m}\Bigr\},
  \label{eq:stakes}
\end{equation}
where a cap is dropped when its denominator is nonpositive and
\[
  \hat v_t=\frac1t\Bigl\{\tfrac14+\sum_{i<t}\bigl(Z_i-\hat\theta^Z_{i-1}\bigr)^2\Bigr\},
  \qquad
  \hat\theta^Z_i=\frac1i\sum_{j\leq i}Z_j,\qquad \hat\theta^Z_0=\hat\theta_0 .
\]
The term $\tfrac14$, the largest variance of a $[0,1]$ variable, keeps early stakes
small before any variance has been observed. The horizon $n^\star$ is the number of
reviews at which the stakes are sharpest. The half-width of the interval is insensitive to $n^\star$, and we use $n^\star=800$ throughout without tuning.

This is what separates the e-process from Serfling's bound. Serfling's width is
set by the largest possible variance, $1/4$; here the width
$\{2\hat v_t\log(2/\alpha)/n\}^{1/2}$ is set by the variance actually observed,
which is $0.147$ for the raw error indicator and far smaller once the predictions
of \eqref{eq:Zt} have learned the rows.


%
\section{Entropy and linear envelopes: construction and proofs}\label{app:entropy}

\subsection{Target and guarantee}\label{ent:app:goal}

We write $\mathbf{p^\star}=\mathbf{p(C^\star)}$ for the corrected composition,
$H^\star=H(\mathbf{p^\star})$ for its entropy, and $\Delta^{L-1}$ for the set of compositions
on the vocabulary $\mathcal L$, with $0\log0=0$. The audit reports a lower end
$L_t$ and an upper end $U_t$ after every audit. We split the total level as
$\alpha=\alpha_-+\alpha_+$, and require
\begin{equation}\label{ent:eq:goal}
  \Pr\bigl\{L_t\le H^\star\le U_t\ \text{for all }t\ge0\bigr\}\ \ge\ 1-\alpha_--\alpha_+
  \qquad\text{for every labeling.}
\end{equation}
Here $\alpha_-$ bounds the probability that the lower end ever exceeds $H^\star$, and
$\alpha_+$ bounds the probability that the upper end ever falls below it. The
experiments use $\alpha_-=\alpha_+=0.0225$.

The diversity interval is $P_2\in[e^{L_t},e^{U_t}]$. We stop once the report is
within 5\% of every value in the interval, which by \cref{ent:prop:report} means
$U_t-L_t\le\log(21/19)$.

\begin{proposition}[Reporting a point value]\label{ent:prop:report}
For $0<\underline P\le\overline P$, the report minimizing
$\max_{P\in[\underline P,\overline P]}|x-P|/P$ is
$x^\star=2\underline P\,\overline P/(\underline P+\overline P)$. Its worst-case
relative error is $(\overline P-\underline P)/(\overline P+\underline P)$, which
equals $\tanh\{(U_t-L_t)/2\}$ for $\underline P=e^{L_t}$ and
$\overline P=e^{U_t}$. So 5\% relative accuracy holds if and only if
$U_t-L_t\le2\,\mathrm{artanh}(0.05)=\log(21/19)$.
\end{proposition}

\begin{proof}
For $x\in[\underline P,\overline P]$ the worst case is at an endpoint, giving
$\max\{x/\underline P-1,\,1-x/\overline P\}$. A report outside the interval is worse
at the nearer endpoint. The first term increases in $x$ and the second decreases;
they meet at $x^\star$ with value $(\overline P-\underline P)/(\overline P+\underline P)$.
Writing $W=U_t-L_t$, this equals $(e^W-1)/(e^W+1)=\tanh(W/2)$, and
$\tanh(W/2)\le0.05$ if and only if $W\le\log(1.05/0.95)$.
\end{proof}

\subsection{The reference composition and the observation from each audit}\label{ent:app:obs}

\paragraph{Gibbs' inequality.}
For a reference $\mathbf r$ with every $r_s>0$, $\mathrm{CE}(\mathbf p,\mathbf r)-H(\mathbf p)=\mathrm{KL}(\mathbf p\|\mathbf r)$.
By Jensen's inequality,
\[
  -\mathrm{KL}(\mathbf p\|\mathbf r)=\sum_{s:p_s>0}p_s\log(r_s/p_s)\le\log\sum_{s:p_s>0}r_s\le0,
\]
with equality if and only if $\mathbf r=\mathbf p$ \citep[Theorem~2.6.3]{cover1991elements}. Since
$\partial_sH(\mathbf r)=-\log r_s-1$, formula \eqref{ent:eq:tangent} gives
$d_s(\mathbf r)=\log(1/r_s)$.

\paragraph{The plug-in composition and the reference.}
Let $O_{t,s}$ be the number of the first $t$ audited records verified as $s$, and let
${q}_{g,t}(s)$ be the predicted share of label $s$ among the unaudited records of
stratum $g$ (the predictions of \cref{sec:linear}). Before audit $t$,
\begin{equation}\label{ent:eq:ref}
  \hat p_{t-1,s}=\frac1N\Bigl(O_{t-1,s}+\sum_gM_{g,t}\,{q}_{g,t}(s)\Bigr),\qquad
  \mathbf{r}_t=(1-\eta)\mathbf{\hat p}_{t-1}+\frac\eta L,\qquad
  \psi_{t,s}=\log\frac1{r_{t,s}},
\end{equation}
with $\eta=0.01$. Four properties follow:
\begin{itemize}
  \item $\sum_s\hat p_{t-1,s}=1$, because $\sum_sO_{t-1,s}=t-1$ and
  $\sum_gM_{g,t}=N-t+1$;
  \item $\mathbf{\hat p}_{t-1}$ equals $\mathbf{p^\star}$ once every record has been audited;
  \item mixing in the uniform distribution gives $\eta/L\le r_{t,s}\le1$, so every
  contribution $\psi_{t,s}$ lies in $[0,\log(L/\eta)]$, even for labels not yet seen;
  \item any predictable reference gives a valid upper end, and the choice of $\mathbf{r}_t$
  affects only the width, through $\mathrm{KL}(\mathbf{p^\star}\|\mathbf{r}_t)$, which is second order
  in the error of $\mathbf{\hat p}_{t-1}$.
\end{itemize}
The predicted contribution of stratum $g$ is
$\hat\mu_{g,t}=\sum_s{q}_{g,t}(s)\,\psi_{t,s}$.

\paragraph{The observations.}
Audit $t$ gives the upper-end observation, which is \eqref{eq:Zt} with every row of
the contribution matrix equal to $\bm{\psi}_t$:
\begin{equation}\label{ent:eq:ZU}
  Z^U_t=\underbrace{\sum_s\hat p_{t-1,s}\,\psi_{t,s}}_{\text{plug-in cross-entropy}}
  +\frac{M_{G_t,t}}{N\pi_{G_t,t}}\bigl(\psi_{t,V_t}-\hat\mu_{G_t,t}\bigr).
\end{equation}
For a candidate composition $\mathbf p$, let $\mathbf{\tilde p}=(N \mathbf p+1)/(N+L)$. The same audit gives
the lower-end observation
\begin{equation}\label{ent:eq:Zp}
  Z_t(\mathbf p)=\frac1N\Bigl(\sum_{i\in \mathcal A_{t-1}}\log\frac1{\tilde p_{v_i}}+\sum_gM_{g,t}\hat\mu_{g,t}\Bigr)
  +\frac{M_{G_t,t}}{N\pi_{G_t,t}}\Bigl(\log\frac1{\tilde p_{V_t}}-\hat\mu_{G_t,t}\Bigr).
\end{equation}
This is \eqref{ent:eq:ZU} with $\log(1/\tilde p_s)$ in place of $\psi_{t,s}$. The
prediction term $\hat\mu_{g,t}$ is kept the same for every candidate, which is what
makes the computation in \cref{ent:app:profile} convex. The smoothing keeps every
contribution finite, at most $\log(N+L)$.

\begin{lemma}[One audit, every cross-entropy]\label{ent:lem:obs}
If $\pi_{g,t}>0$ whenever $M_{g,t}>0$, then
$\mathbb E[Z^U_t\mid\mathcal F_{t-1}]=\mathrm{CE}(\mathbf{p^\star},\mathbf{r}_t)=H^\star+\mathrm{KL}(\mathbf{p^\star}\|\mathbf{r}_t)\ge H^\star$,
and $\mathbb E[Z_t(\mathbf p)\mid\mathcal F_{t-1}]=\mathrm{CE}(\mathbf{p^\star},\mathbf{\tilde p})$ for every
composition $\mathbf p$.
\end{lemma}

\begin{proof}
Let $\bm  \vartheta\in(0,1]^L$ be predictable, let the $\hat\mu_{g,t}$ be any predictable
numbers, and set
\[
  Y_t=\frac1N\Bigl(\sum_{i\in \mathcal A_{t-1}}\log\frac1{\vartheta_{v_i}}+\sum_gM_{g,t}\hat\mu_{g,t}\Bigr)
  +\frac{M_{G_t,t}}{N\pi_{G_t,t}}\Bigl(\log\frac1{\vartheta_{V_t}}-\hat\mu_{G_t,t}\Bigr).
\]
Given $\mathcal F_{t-1}$ and $G_t=g$, the audited record is uniform over the
$M_{g,t}$ unaudited records of $g$. Averaging over $G_t$ cancels $\pi_{g,t}$ exactly
as in the proof of \cref{lem:unbiased}, and the $\hat\mu$ terms cancel. So
$\mathbb E[Y_t\mid\mathcal F_{t-1}]=N^{-1}\sum_{i=1}^N\log(1/\vartheta_{v_i})=\mathrm{CE}(\mathbf{p^\star},\bm \vartheta)$.
The argument uses only that $\hat\mu_{g,t}$ is predictable.

With $\bm \vartheta=\mathbf{r}_t$ we get $Y_t=Z^U_t$, because \eqref{ent:eq:ref} gives
\[
  \frac1N\Bigl(\sum_{i\in \mathcal A_{t-1}}\psi_{t,v_i}+\sum_gM_{g,t}\hat\mu_{g,t}\Bigr)=\sum_s\hat p_{t-1,s}\psi_{t,s}.
\]
With $\bm \vartheta=\mathbf{\tilde p}$ we get $Y_t=Z_t(\mathbf p)$.
\end{proof}

\subsection{The upper end}\label{ent:app:upper}

For each candidate value $h$, the upper-end e-process is
\[
  E^U_t(h)=\prod_{i\le t}\bigl\{1-\lambda^U_i\,(Z^U_i-h)\bigr\},
\]
the downward e-process of \cref{sec:linear} applied to \eqref{ent:eq:ZU}. The stake
$\lambda^U_i\ge0$ is predictable and does not depend on $h$.

\paragraph{Range and cap.}
Before audit $t$, everything in \eqref{ent:eq:ZU} is known except the stratum drawn
and the label revealed. The stratum can be any $g$ with $M_{g,t}>0$, and the label
can be any label in $\mathcal L$, including one not yet seen. So the largest possible
value is
\begin{equation}\label{ent:eq:range-upper}
  \overline Z^U_t=\sum_s\hat p_{t-1,s}\psi_{t,s}
  +\max_{g:\,M_{g,t}>0}\frac{M_{g,t}}{N\pi_{g,t}}\Bigl(\max_s\psi_{t,s}-\hat\mu_{g,t}\Bigr).
\end{equation}
Let $h^{\min}_t$ be the smallest candidate tested at audit $t$: $0$ in the basic
version and last lower endpoint $L_{t-1}$ with cross-capping (\cref{ent:prop:cross}). We require
\begin{equation}\label{ent:eq:cap-upper}
  0\le\lambda^U_t\le\frac{0.9}{\overline Z^U_t-h^{\min}_t},
\end{equation}
dropping the cap when its denominator is nonpositive.

\begin{proposition}[Upper confidence sequence]\label{ent:prop:upper}
Let $U_t$ be the smallest $h\ge h^{\min}_t$ with $E^U_s(h)\ge1/\alpha_+$ for some
$s\le t$, or $U_t=\log L$ if there is none, with $h^{\min}_t=0$. Then
$\Pr\{H^\star\le U_t\ \forall t\}\ge1-\alpha_+$ for every labeling.
\end{proposition}

\begin{proof}
\emph{Positivity.} For $h\ge h^{\min}_i$, the cap \eqref{ent:eq:cap-upper} gives
$1-\lambda^U_i(Z^U_i-h)\ge1-\lambda^U_i(\overline Z^U_i-h^{\min}_i)\ge0.1$.

\emph{Supermartingale.} At $h=H^\star\ge0$, \cref{ent:lem:obs} gives each factor
conditional mean $1-\lambda^U_i\,\mathrm{KL}(\mathbf{p^\star}\|\mathbf {r_i})\le1$. So $E^U_t(H^\star)$
is a nonnegative supermartingale starting at one. By Ville's inequality
(\cref{lem:ville}) it ever reaches $1/\alpha_+$ with probability at most $\alpha_+$.

\emph{Monotonicity.} Each factor is positive and nondecreasing in $h\ge0$. So if some
$h<H^\star$ were ruled out at audit $s$, then $E^U_s(H^\star)\ge E^U_s(h)\ge1/\alpha_+$.

\emph{Start.} $\log L\ge H(\mathbf p)$ for every composition $\mathbf p$.
\end{proof}

\subsection{The lower end}\label{ent:app:lower}

\paragraph{Smoothing.}
For any composition $\mathbf p$, $1/(N+L)\le\tilde p_s\le1$. At the truth:
\begin{itemize}
  \item if $p^\star_s>0$, then $\tilde p^\star_s\ge p^\star_s/(1+L/N)$, so
  $\log(1/\tilde p^\star_s)\le\log(1/p^\star_s)+c$ with $c=\log(1+L/N)$;
  \item if $p^\star_s=0$, no record is verified as $s$, so label $s$ never enters
  $Z_t(\mathbf p)$.
\end{itemize}
By \cref{ent:lem:obs},
\begin{equation}\label{ent:eq:smooth-mean}
  \mathbb E\bigl[Z_t(\mathbf{p^\star})\bigm|\mathcal F_{t-1}\bigr]\le H^\star+c .
\end{equation}
On the Kenyan data, $c=\log(1+50/33{,}403)=0.0015$ nats.

\paragraph{Candidates.}
Every audited record verified as $s$ belongs to $p^\star_s$, so
$\mathbf{p^\star}\in\mathcal P_t=\{\mathbf p\in\Delta^{L-1}:p_s\ge O_{t,s}/N\}$ for every $t$, whatever
happens in the audit.

\paragraph{The e-processes and the lower end.}
With a predictable stake $\lambda^L_i\ge0$ shared by all candidates and all $h$, let
\[
  E^L_t(\mathbf p;h)=\prod_{i\le t}\bigl\{1+\lambda^L_i\bigl(Z_i(\mathbf p)-h-c\bigr)\bigr\},\qquad
  F_t(\mathbf p;h)=\log E^L_t(\mathbf p;h).
\]
At audit $t$, a value $h$ is ruled out if $\min_{\mathbf p\in\mathcal P_t}F_t(\mathbf p;h)\ge\log(1/\alpha_-)$,
or if a certified lower bound on this minimum is (\cref{ent:app:profile}). The lower
end $L_t$ is the largest value ruled out so far, and $L_0=0$.

\paragraph{How $Z_i(\mathbf p)$ depends on $\mathbf p$.}
Collecting terms in \eqref{ent:eq:Zp},
\begin{equation}\label{ent:eq:Zp-linear}
  Z_i(\mathbf p)=\sum_s\chi_{i,s}\log\frac1{\tilde p_s}+\varrho_i,\qquad
  \chi_{i,s}=\frac1N\Bigl(O_{i-1,s}+\frac{M_{G_i,i}}{\pi_{G_i,i}}\mathbf 1\{V_i=s\}\Bigr)\ge0,
\end{equation}
where $\varrho_i=N^{-1}(\sum_gM_{g,i}\hat\mu_{g,i}-M_{G_i,i}\hat\mu_{G_i,i}/\pi_{G_i,i})$
does not depend on $\mathbf p$, and $\sum_s\chi_{i,s}=N^{-1}(i-1+M_{G_i,i}/\pi_{G_i,i})$.

\paragraph{The cap.}
Let $h^{\max}_t$ be the largest candidate tested at audit $t$: $\log L$ in the basic
version and last upper endpoint $U_{t-1}$ with cross-capping. We require
$\lambda^L_t(\sum_s\chi_{t,s}-\varrho_t+h^{\max}_t+c)\le1$ for every stratum that can
be drawn; the label revealed does not affect the left side. Substituting,
\begin{equation}\label{ent:eq:lower-cap}
  \lambda^L_t\Bigl[\frac{t-1}N+\frac{M_{g,t}}{N\pi_{g,t}}\bigl(1+\hat\mu_{g,t}\bigr)-\frac1N\sum_kM_{k,t}\hat\mu_{k,t}+h^{\max}_t+c\Bigr]\le1
  \quad\text{for every }g\text{ with }M_{g,t}>0 .
\end{equation}
The $\pi_t$-weighted average of these brackets is at least $1+h^{\max}_t+c$. It grows when a large stratum, or one with
a large predicted contribution, is drawn with small probability. Priority-aware
sampling does exactly this for large, well-known strata, which is why the
allocation in \cref{ent:app:alloc} includes a floor.

\begin{theorem}[Lower confidence sequence]\label{ent:thm:lower}
Suppose each $\lambda^L_i$ is predictable and satisfies \eqref{ent:eq:lower-cap} with
$h^{\max}_i=\log L$. Then $\Pr\{L_t\le H^\star\ \forall t\}\ge1-\alpha_-$ for every
labeling. Together with \cref{ent:prop:upper}, the interval $[L_t,U_t]$
satisfies \eqref{ent:eq:goal}.
\end{theorem}

\begin{proof}
Fix $i$, $h\le\log L$ and a composition $p$. Write $\chi_s=\chi_{i,s}$,
$\lambda=\lambda^L_i$, $\varsigma=1+\lambda(\varrho_i-h-c)$ and
$J(p)=\sum_s\chi_s\log(1/\tilde p_s)\ge0$, so that the factor is
$\varsigma+\lambda J(p)$.

\emph{Positivity.} By \eqref{ent:eq:lower-cap}, $\varsigma\ge\lambda\sum_s\chi_s>0$
when $\lambda>0$, so every factor is positive.

\emph{Monotonicity.} Each factor decreases in $h$, and the stake does not depend on
$h$. So $F_t(p;h)$ is nonincreasing in $h\le\log L$.

\emph{Validity.} By \eqref{ent:eq:smooth-mean}, at $\mathbf p=\mathbf{p^\star}$ and $h=H^\star$ every
factor has conditional mean at most one. So $E^L_t(\mathbf{p^\star};H^\star)$ is a nonnegative
supermartingale starting at one, and by \cref{lem:ville} it ever reaches $1/\alpha_-$
with probability at most $\alpha_-$.

Suppose $L_t>H^\star$. Then some $h\in(H^\star,\log L]$ was ruled out at some audit
$s\le t$, by a statistic no larger than $\min_{\mathcal P_s}F_s(\cdot;h)$. Since
$\mathbf{p^\star}\in\mathcal P_s$ and $F_s$ is nonincreasing in $h$,
$\log(1/\alpha_-)\le F_s(\mathbf{p^\star};h)\le F_s(\mathbf{p^\star};H^\star)$. This has probability at
most $\alpha_-$.

\emph{Two sides.} A union bound with \cref{ent:prop:upper} gives \eqref{ent:eq:goal}.
\end{proof}

\subsection{Computing the lower end: the convex profile}\label{ent:app:profile}

Ruling out $h$ requires the minimum of $F_t(\cdot;h)$ over $\mathcal P_t$. The cap
that keeps the factors positive also makes this a convex problem.

\begin{proposition}[Convex profile]\label{ent:prop:profile}
Under \eqref{ent:eq:lower-cap}, $\mathbf p\mapsto F_t(\mathbf p;h)$ is convex on $\Delta^{L-1}$ for
every $h\le h^{\max}_t$. For any trial point $\mathbf p^\circ\in\mathcal P_t$, the value
$\underline F_t(h)$ in \eqref{ent:eq:plane} is a lower bound on
$\min_{\mathcal P_t}F_t(\cdot;h)$. Using it in place of the minimum preserves
\cref{ent:thm:lower}.
\end{proposition}

\begin{proof}
Keep the notation of the previous proof and let $\upsilon_s(\mathbf p)=N/(Np_s+1)$. The
derivatives of $\log(1/\tilde p_s)$ in $p_s$ are $-\upsilon_s$ and $\upsilon_s^2$,
with no cross terms. For $\zeta(\mathbf p)=\log\{\varsigma+\lambda J(\mathbf p)\}$ and any
direction $\mathbf y$,
\[
  \mathbf y^\top\bm \nabla^2\zeta\,y=\frac{\lambda\sum_s\chi_s\upsilon_s^2y_s^2}{\varsigma+\lambda J}-\frac{\lambda^2\bigl(\sum_s\chi_s\upsilon_sy_s\bigr)^2}{(\varsigma+\lambda J)^2}.
\]
By Cauchy--Schwarz,
$(\sum_s\chi_s\upsilon_sy_s)^2\le(\sum_s\chi_s)(\sum_s\chi_s\upsilon_s^2y_s^2)$. So
this is nonnegative whenever $\varsigma+\lambda J\ge\lambda\sum_s\chi_s$, which the
cap gives. Summing over $i$ preserves convexity.

The lower bound is the supporting-plane inequality below. The proof of
\cref{ent:thm:lower} never uses how the minimum was computed.
\end{proof}

\paragraph{The supporting-plane bound.}
For a convex differentiable function,
$F(\mathbf p)\ge F(\mathbf p^\circ)+\nabla F(\mathbf p^\circ)^\top(\mathbf p- \mathbf p^\circ)$ \citep{vandenberghe2004convex}.
Over $\mathcal P_t$, the linear part is smallest with every coordinate at its lower
limit $O_{t,s}/N$ and the free mass $1-t/N$ on the label with the smallest gradient
coordinate. Writing $\bm \nabla=\nabla F_t(\mathbf p^\circ;h)$,
\begin{equation}\label{ent:eq:plane}
  \underline F_t(h)=F_t(\mathbf p^\circ;h)-\bm \nabla^\top \mathbf p^\circ+\frac1N\sum_s\nabla_sO_{t,s}+\Bigl(1-\frac tN\Bigr)\min_s\nabla_s\ \le\ \min_{p\in\mathcal P_t}F_t(\mathbf p;h),
\end{equation}
where
\[
  \partial_sF_t(\mathbf p;h)=-\sum_{i\le t}\frac{\lambda^L_i\,\chi_{i,s}\,\upsilon_s(\mathbf p)}{1+\lambda^L_i\{Z_i(\mathbf p)-h-c\}} .
\]

\paragraph{In practice.}
We take $\mathbf p^\circ$ from projected gradient descent on $\mathcal P_t$, with at most 60
iterations, and rule out $h$ when $\underline F_t(h)\ge\log(1/\alpha_-)$. The gap
$F_t(\mathbf p^\circ;h)-\underline F_t(h)$ is lost evidence, not lost validity, so an
inexact optimizer can only widen the interval.

If several stake sequences each satisfy \eqref{ent:eq:lower-cap}, the logarithm of
their weighted average, $\log\sum_jw_j\exp F^{(j)}_t$ with $w_j\ge0$ and
$\sum_jw_j=1$, is again convex in $\mathbf p$ and nonincreasing in $h$. This holds because
log-sum-exp is convex and increasing in each argument. Mixtures of stake scales are
therefore allowed.

\subsection{Allocation and cross-capped stakes}\label{ent:app:alloc}

\paragraph{Allocation.}
Strata are submitted label $\times$ camera station (300 strata). In our experiment, PAS draws stratum $g$
with probability
\begin{equation}\label{ent:eq:alloc}
  \pi_{g,t}=0.4\frac{M_{g,t}}{\sum_kM_{k,t}}
  +0.3\frac{M_{g,t}(1+\hat\mu_{g,t})}{\sum_kM_{k,t}(1+\hat\mu_{k,t})}
  +0.3\frac{M_{g,t}\hat\sigma_{g,t}}{\sum_kM_{k,t}\hat\sigma_{k,t}},
  \qquad
  \hat\sigma_{g,t}^2=\sum_s q_{g,t}(s)\bigl(\psi_{t,s}-\hat\mu_{g,t}\bigr)^2 .
\end{equation}
\begin{itemize}
  \item \emph{The priority.} The last term is the priority of
  \cref{prop:neyman} applied to the influence $\bm{\psi}_t$.
  \item \emph{The floor.} The middle term implies
  $M_{g,t}(1+\hat\mu_{g,t})/\pi_{g,t}\le\sum_kM_{k,t}(1+\hat\mu_{k,t})/0.3$ for every
  $g$, which bounds the stratum-dependent term of \eqref{ent:eq:lower-cap}.
  \item \emph{The baseline.} SRS sets $\pi_{g,t}=M_{g,t}/\sum_kM_{k,t}$.
\end{itemize}
Every nonempty stratum has positive probability, so \cref{lem:unbiased} applies.

\begin{proposition}\label{ent:prop:cross}
Run both ends on the same audits with $h^{\min}_t=L_{t-1}$ in
\eqref{ent:eq:cap-upper} and $h^{\max}_t=U_{t-1}$ in \eqref{ent:eq:lower-cap},
testing only candidates in $[L_{t-1},U_{t-1}]$. Then \eqref{ent:eq:goal} holds.
\end{proposition}

\begin{proof}
The endpoints only move inward, and they start at $0\le H^\star\le\log L$. Let
$t_0$ be the first audit after which $H^\star\notin[L_{t_0},U_{t_0}]$. The
modified stakes
\[
  \lambda^{U\prime}_i=\lambda^U_i\mathbf 1\{H^\star\ge L_{i-1}\},\qquad
  \lambda^{L\prime}_i=\lambda^L_i\mathbf 1\{H^\star\le U_{i-1}\}
\]
are predictable, since $H^\star$ is a fixed number. With them, the factors at the
truth are positive and have conditional mean at most one, so both modified processes
are nonnegative supermartingales. They reach $1/\alpha_+$ and $1/\alpha_-$ with
probabilities at most $\alpha_+$ and $\alpha_-$.

For $i\le t_0$, $H^\star\in[L_{i-1},U_{i-1}]$, so the modified stakes equal the
actual ones up to audit $t_0$. A failure at $t_0$ is one of two events:
\begin{itemize}
  \item some $h\in[L_{t_0-1},H^\star)$ has $E^U_{t_0}(h)\ge1/\alpha_+$, which forces
  $E^U_{t_0}(H^\star)\ge1/\alpha_+$;
  \item some $h\in(H^\star,U_{t_0-1}]$ is ruled out by the lower end, which forces
  $F_{t_0}(\mathbf{p^\star};H^\star)\ge\log(1/\alpha_-)$.
\end{itemize}
Either is a crossing of a modified process, and a union bound finishes the proof.
\end{proof}

The proof uses only that the ceiling $h^{\max}_t$ is predictable and never increases,
not that it is an upper confidence bound. So a one-sided lower bound can use any such
ceiling, for example one set near the current plug-in entropy.

\section{Scores beyond the biodiversity pillars}
\label{app:scores}

\paragraph{Ratios.}
Recall and the F1 score are both ratios of two linear scores. For each
submitted label \(r\) and verified label \(s\), let \(U_{rs}\) be the
contribution of one such record to the numerator and \(W_{rs}\) its
contribution to the denominator. Since \(C^\star_{rs}\) is the number
of records with that label pair, the score is
\[
  \rho^\star
  = \frac{\sum_{r,s} U_{rs} C^\star_{rs}}
         {\sum_{r,s} W_{rs} C^\star_{rs}},
\]
where we assume the denominator is positive. For a candidate ratio $\rho$,
let $D_{rs}=U_{rs}-\rho W_{rs}$. Then $\rho^\star=\rho$ exactly when
the linear score $N^{-1}\sum_{r,s}D_{rs}C^\star_{rs}$ equals zero.
The entries of $D$ may be negative, but the proof of Theorem~2 applies
unchanged: their known range gives predictable bounds on $Z_t$, so we
can choose the stakes to keep both betting factors nonnegative.
We therefore use the two e-processes in~(6) at $m=0$, one for detecting
a positive score and the other for detecting a negative score. At each
audit, we retain every $\rho$ for which neither process has reached
$2/\alpha$. The retained values form a $1-\alpha$ confidence sequence
for $\rho^\star$.

\paragraph{Bounded curvature.}
If $\nabla^2g$ is bounded by $\gamma_g$ on the relevant region, then
$g(\mathbf{p^\star})=g(\mathbf r)+\nabla g(\mathbf r)^\top(\mathbf{p^\star}-\mathbf r)+R$ with $|R|\leq\frac{\gamma_g}2\|\mathbf{p^\star}-\mathbf r\|^2$;
a linear confidence sequence for the first-order term and a composition region for
the remainder, at a separate budget, certify $g$. Shannon entropy is not of this
type, because $\partial^2H/\partial p_s^2=-1/p_s$ is unbounded near zero.

\paragraph{Support-based scores.}
Richness changes by one when a single record reveals a new species. If a procedure
declares an eligible species absent after $T$ reviews and errs with probability at
most $\alpha$ under every alternative in which one record carries it, then
$\E[T]\geq(1-\alpha)N$: couple with each single-record alternative, which agrees
until that record is reviewed, and sum over records. Presence, by contrast, is
certified by one verified record.